\documentclass[11pt,a4paper]{article}
\usepackage{calligra}

\usepackage[T1]{fontenc}
\usepackage[normalem]{ulem}
\usepackage{placeins}
\usepackage{tikz}
\usepackage{amsmath,amsthm}
\usepackage{amssymb}
\usepackage{graphicx}
\usepackage{dsfont}
\usepackage{fullpage}
\usepackage{multirow}
\usepackage[affil-it]{authblk}
\usepackage{braket}

\usepackage{enumerate}
\usepackage[shortlabels]{enumitem}

\usepackage{mathtools}
\usepackage{nicefrac}

\usepackage[backend=biber,doi=false,url=false,maxbibnames=5,style=alphabetic]{biblatex}
\renewbibmacro{in:}{}
\usepackage[toc,page]{appendix}
\usepackage{tocloft}
\definecolor{linkblue}{HTML}{001487}

\usepackage[colorlinks, allcolors=linkblue, breaklinks=true]{hyperref}
\newtheorem{theorem}{Theorem}[section]

\newtheorem{proposition}[theorem]{Proposition}
\newtheorem{corollary}[theorem]{Corollary}
\theoremstyle{remark}

\newtheorem*{remark}{Remark}

\usepackage{cleveref}

\usepackage{accents}

\newcommand{\cH}{\mathcal{H}}
\newcommand{\cD}{\mathcal{D}}
\newcommand{\D}{\mathcal{D}}
\newcommand{\Dpure}{\D_\text{pure}}

\DeclareMathAlphabet{\mathpzc}{OT1}{pzc}{m}{it}
 \newcommand{\cs}{\mathpzc{s}}

\newcommand{\cP}{\mathcal{P}}

\newcommand{\cL}{\mathcal{L}}

\newcommand{\cB}{\mathcal{B}}
\newcommand{\Bsa}{\mathcal{B}_\text{sa}}

\newcommand{\bE}{\mathbb{E}}

\newcommand{\one}{\mathds{1}}
\newcommand{\eps}{\varepsilon}

\newcommand{\RR}{\mathbb{R}}
\newcommand{\R}{\mathbb{R}}

\newcommand{\gen}{\mathcal{L}}

\DeclareMathOperator{\tr}{Tr}

\DeclareMathOperator{\id}{id}

\DeclareMathOperator{\bin}{bin}

\DeclareMathOperator{\Eig}{Eig}
\DeclareMathOperator{\diag}{diag}

\renewcommand{\epsilon}{\varepsilon}
\newcommand{\Be}{B_\varepsilon}

\renewcommand{\one}{\mathds{1}}

\DeclareMathOperator{\spec}{spec}

\newcommand{\inv}{^{-1}}

\renewcommand{\d}{\operatorname{d}\!}

\newcommand{\mmm}{\mathcal{M}} 

\newcommand{\e}{\mathrm{e}}
\DeclareMathOperator{\TV}{TV}
\newcommand{\typeone}{Concave-Type}
\newcommand{\typetwo}{Convex-Type}

\newcommand{\unorm}[1]{{\left\vert\kern-0.25ex\left\vert\kern-0.25ex\left\vert #1 
    \right\vert\kern-0.25ex\right\vert\kern-0.25ex\right\vert}}
\begin{document}
\title{Universal proofs of entropic continuity bounds\\ via majorization flow}

\author{Eric P. Hanson\thanks{Email: \texttt{ericpatrickhanson@gmail.com}} }
\author{Nilanjana Datta\thanks{Email: \texttt{n.datta@damtp.cam.ac.uk}}}
\affil{\small Department of Applied Mathematics and Theoretical Physics\\University of Cambridge}
\date{July 22, 2021}

\maketitle

\begin{abstract}
We introduce a notion of \emph{majorization flow}, and demonstrate it to be a powerful tool for deriving simple and universal proofs of continuity bounds for Schur concave functions with a particular emphasis on entropic functions relevant in information theory. In particular, for the case of the $\alpha$-R\'enyi entropy, whose connections to thermodynamics are discussed in this article, majorization flow yields a Lipschitz continuity bound for the case $\alpha >1$, thus resolving an open problem and providing a substantial improvement over previously known bounds. 

\end{abstract}

\section{Introduction}

Majorization is an ordering of vectors\footnote{Technically, it is a \emph{pre-order}.} which provides a means to describe one vector as being more disordered, less disordered, or incomparable with another vector. Majorization has natural connections to entropies, which quantify disorder or randomness of probability distributions or quantum states, and thus provide a measure of information encoded in a random variable or a quantum state. In fact, almost every entropy considered in the information-theoretic literature is Schur concave, meaning if one vector is more disordered than another, according to the majorization order, then it has higher entropy than the other. Perhaps as a consequence, majorization has proven to be a very useful tool in quantum and classical information theory.

In quantum information theory, majorization plays a key role in the theory of bipartite pure state entanglement due to Nielsen's theorem \cite{Nielsen-LOCC}. This result states that given two pure states $\psi_{AB}$ and $\phi_{AB}$ of a bipartite system $AB$, the state $\psi_{AB}$ can be transformed into $\phi_{AB}$ via local operations and classical communication if and only if the reduced state of $\psi_{AB}$ for a subsystem ($A$ or $B$) is majorized by the corresponding reduced state of $\phi_{AB}$.

Entropies provide a way to quantify uncertainty, and hence play a vital role in information theory. In classical information theory, the Shannon entropy gives the data compression limit of a source, while the analogous von Neumann entropy gives the data compression limit of a quantum information source. The classical (resp. quantum) R\'enyi entropies \cite{Ren61} and Tsallis entropies \cite{Tsallis1988} constitute families of entropies which generalize the Shannon entropy (resp. von Neumann entropy). In quantum information theory, entropies are also used to quantify entanglement, e.g.~the \emph{entanglement entropy} of a bipartite pure state is the entropy of one of its marginals.

Recently, it has been shown that the majorization order on probability vectors behaves well with respect to the total variation distance, in the sense that given any probability vector $p$ and $\eps > 0$, there exists a {\em{minimal}} and {\em{maximal probability vector}} in majorization order in the $\eps$-ball around $p$ with respect to the total variation distance; both these probability vectors were independently found in \cite{HD17,HOS18}, and in fact the minimal vector was already found in \cite{HY10}). In the quantum setting, this provides a natural connection between majorization and trace distance (and between majorization and the total variation distance in the classical setting). This connection was exploited to establish local entropic continuity bounds in \cite{HD17}. In this work, we develop a notion of \emph{majorization flow}, which traces out the path of the minimizer in majorization order over the $\eps$-ball as $\eps$ is changed infinitesimally.

Majorization flow turns out to be a very useful tool for obtaining remarkably simple and universal proofs of continuity bounds for numerous well-known families of entropies, including R\'enyi entropies \cite{Ren61}, Tsallis entropies \cite{Tsallis1988}, the so-called unified entropies \cite{RT91}, entropies induced by $f$-divergences \cite{Pet85}, and the concurrence \cite{Woo01}. In particular, it allows us to establish {\em{Lipschitz continuity bounds}} for the $\alpha$-R\'enyi entropy with $\alpha > 1$. This resolves a longstanding open problem, which was also presented at the Open Problem Session of the workshop \emph{2017 Beyond I.I.D. in Information Theory} (held in Singapore), and provides a substantial improvement over previously known bounds.

\paragraph{Overview and summary of contributions}
\begin{itemize}
	\item In \Cref{sec:majflow}, we introduce the notion of \emph{majorization flow}, and show how it can yield short, simple and universal proofs of continuity bounds for entropic functions. In \Cref{thm:integral-formulas} we also introduce two important relations \eqref{eq:mmm-integral-formula} and \eqref{eq:H-integral-formula} satisfied by the majorization flow. The first describes the path traced out by the majorization flow in terms of its generator, while the second gives an integral formula for the change of an entropic function in terms of a directional derivative along the flow. We also provide a necessary and sufficient condition for a Schur concave function which is continuously differentiable on the interior of the probability simplex to be Lipschitz continuous on the whole probability simplex in \Cref{cor:opt-Lip}.
	\item In \Cref{sec:q-to-c}, we show the equivalence between continuity bounds for quantum states and for classical probability distributions for single-partite entropies.
	\item In \Cref{sec:gen-and-proofs}, we describe the generator of the majorization flow with respect to total variation distance, and prove \Cref{thm:integral-formulas}.
	\item In \Cref{sec:cty-from-majflow}, we recall the definition of a large family of entropies called the $(h,\phi)$-entropies \cite{SMMP93a}, which include most of the known entropic quantities which are of relevance in information theory. We  consider two subclasses of $(h,\phi)$-entropies, namely \typeone{} $(h,\phi)$-entropies and \typetwo{} $(h,\phi)$-entropies. In \Cref{thm:hphi-GCB}, we establish tight uniform continuity bounds for all \typeone{} $(h,\phi)$-entropies by using majorization flow, and provide necessary and sufficient conditions for such an entropy to be Lipschitz continuous in \Cref{prop:typ1-Lip}. The latter proposition also establishes that all {\em{smoothed}} \typeone{} $(h,\phi)$-entropies are Lipschitz continuous. In \Cref{thm:Convex-type-Lip}, we state necessary and sufficient conditions for \typetwo{} $(h,\phi)$-entropies to be Lipschitz continuous, and  provide bounds on their optimal Lipschitz constant.
	
	In \Cref{sec:tsallis-renyi}, we apply these results to Tsallis and R\'enyi entropies. In \Cref{cor:Tsallis-Lipschitz}, we prove that $\alpha$-Tsallis entropies are Lipschitz continuous if and only if $\alpha>1$, and determine their optimal Lipschitz constant in the latter case. In \Cref{cor:Renyi-lipschitz}, we prove that the $\alpha$-R\'enyi entropy is Lipschitz continuous if and only if $\alpha > 1$, and establish bounds on its optimal Lipschitz constant in the latter case. This provides a continuity bound on the $\alpha$-R\'enyi entropy for $\alpha > 1$ which is much tighter than previously known bounds in the literature. In \Cref{sec:discuss-prev-bounds}, we discuss why the previous techniques fail to provide good continuity bounds on the $\alpha$-R\'enyi entropies for large $\alpha$ or  large dimension $d$, and how majorization flow gives insight into the previously-used techniques. In \Cref{sec:connect-thermo}, we discuss the connection between $\alpha$-R\'enyi entropies and thermodynamic free energies (which was introduced by \cite{Baez11}), and interpret our continuity bounds for R\'enyi entropies through the lens of free energies.
	
	In \Cref{sec:more-cty-bounds}, we apply the results of \Cref{sec:cty-from-majflow} to the von Neumann entropy (or, equivalently in this case, to the Shannon entropy), the so-called $(s,\alpha)$-unified entropies, entropies induced by $f$-divergences, and the concurrence.

	\item In \Cref{sec:butterflies} we show that the cumulative probability distribution for the number of distinct elements obtained from $N$ i.i.d. samples of $\{1,\dotsc,M\}$ is Lipschitz continuous in the underlying probability distribution, and obtain a tight uniform continuity bound on the expected number of distinct elements. In particular, we show the optimal Lipschitz constant for the expected number of distinct elements is given simply by $N$, and in particular is independent of $M$.
\end{itemize}

\section{Notation and definitions} \label{sec:notation}
In this article, we consider both probability vectors $p\in \cP \subset\R^d$,
\[
\cP := \left\{ p = (p_1,\dotsc,p_d)\in\R^d : p_i \geq 0 \text{ for }i=1,\dotsc,d,\, \sum_{i=1}^d p_i = 1 \right\}
\]
as well as quantum states $\rho \in \cD(\cH)$ where $\cH$ is a $d$-dimensional Hilbert space, and
\[
\cD(\cH) = \left\{ \rho \in \cB(\cH): \rho \geq 0, \,\tr(\rho) = 1 \right\},
\]
where $\cB(\cH)$ is the set of operators on $\cH$. As discussed in \Cref{sec:q-to-c}, for the continuity bounds discussed in this article, the two frameworks are equivalent. The extremal elements of $\cP$ are permutations of the probability vector $\psi := (1,0,\dotsc,0)$. The extremal elements of $\cD(\cH)$ are rank-1 projections, and are called \emph{pure states}.

We denote the completely mixed state by $\tau := \frac{\one}{d} \in \cD(\cH)$, and the analogous uniform distribution by $u := (\frac{1}{d},\dotsc, \frac{1}{d})\in \cP$. A pure state is a rank-1 density matrix; we denote the set of pure states by $\Dpure(\cH)$. For two quantum states $\rho,\sigma \in \D(\cH)$, the \emph{trace distance} between them 
is given by
\[
T(\rho,\sigma) = \frac{1}{2}\|\rho-\sigma\|_1
\]
where $\|A\|_1 = \tr|A|$ for $A \in \cB(\cH)$.
The trace distance $T(\rho,\sigma)$ has an operational interpretation in terms of the optimal success probability $p$ in distinguishing between two quantum states $\rho$ and $\sigma$ by a 2-outcome POVM:
\[
p = \frac{1}{2}(1 + T(\rho, \sigma)).
\]
Hence, the trace distance can be seen as a measure of indistinguishability between $\rho$ and $\sigma$. Analogously, the total variation distance between $p, q\in \cP$ is defined as
\[
\TV(p,q) = \frac{1}{2}\|p-q\|_1
\]
and is also endowed with an interpretation in terms of distinguishability.

Recall that a function $F: \cD(\cH)\to \mathbb{R}$ is $k$-\emph{Lipschitz} (with respect to the trace distance) if for all $\rho,\sigma \in \cD(\cH)$,
\[
|F(\rho) - F(\sigma)| \leq k\, T(\rho,\sigma).
\]
The smallest $k>0$ such that $F$ is $k$-Lipschitz is called the \emph{optimal Lipschitz constant} for $F$. The function $F$ is said to be \emph{Lipschitz} continuous if it is $k$-Lipschitz for some $k>0$.

For $\eps > 0$, we define the $\eps$-ball (in trace distance) around $\sigma\in \D(\cH)$ as the set
\begin{equation}\label{eq:eps-ball}
\Be(\sigma) = \{ \omega\in \D(\cH): T(\omega,\sigma) \leq \eps \}, 
\end{equation}
and likewise the $\eps$-ball (in total variation distance) around a probability vector $p\in \cP$ as the set
\begin{equation} \label{eq:eps-ball-TV}
\Be(p) = \{ q\in \cP: \TV(p,q) \leq \eps \}.
\end{equation}
For any $A\in \Bsa(\cH)$, let $\lambda_+ (A)$ and $\lambda_-(A)$ denote the maximum and minimum eigenvalue of $A$, respectively, and $k_+(A)$ and $k_-(A)$ denote their multiplicities. Let $\lambda_j(A)$ denote the $j$th largest eigenvalue, counting multiplicity; that is, the $j^{th}$ element of the ordering
\[
\lambda_1(A)\geq \lambda_2(A) \geq \dotsm \geq \lambda_d(A).
\]
We set $\vec \lambda (A) := (\lambda_i(A))_{i=1}^d \in\R^d$ and denote the spectrum of $A\in \Bsa(\cH)$ (i.e.\@~its set of eigenvalues) by $\spec A \subset \R$.

The set of probability vectors with strictly positive entries is denoted $\cP_+$. For a vector $r\in \R^d$, $r_+$ denotes its largest entry, and $r_-$ denotes its smallest entry.
We use $\log x$ for the base-2 logarithm of $x$ and $\ln x$ for the natural logarithm of $x$.
\paragraph{Majorization of vectors}
Given $x\in \R^d$, write $x^\downarrow = (x^\downarrow_j)_{j=1}^d$ for the permutation of $x$ such that $x^\downarrow_1 \geq x^\downarrow_2 \geq \dotsm \geq x^\downarrow_d$. For $x,y\in \R^d$, we say $x$ \emph{majorizes} $y$, written $x \succ y$, if 
	\begin{equation} \label{def:majorize}
	 \sum_{j=1}^k x^\downarrow_j \geq \sum_{j=1}^k y^\downarrow_j \quad \forall k=1,\dotsc,d-1, \quad \text{and}\quad \sum_{j=1}^d x^\downarrow_j = \sum_{j=1}^d y^\downarrow_j.
	 \end{equation}

We say a function $\varphi: \cP \to \R$ is Schur convex on a set $S\subset \cP$ if  for $p,q\in S$, $p\prec q\implies \varphi(p) \leq \varphi(q)$. If $S = \cP$, we simply say $\varphi$ is Schur convex. We say $\varphi$ is \emph{Schur concave} on $S$ if $-\varphi$ is Schur convex on $S$, and likewise, $\varphi$ is Schur concave if $-\varphi$ is Schur convex. One useful characterization of Schur convex functions is if $\varphi : \cP \to \R$ is differentiable and symmetric, then it is Schur convex if and only if
\begin{equation} \label{eq:S-convex-condition}
(p_i - p_j) \left[ \partial_{p_i}\varphi(p) - \partial_{p_j} \varphi(p) \right] \geq 0 \qquad \forall i,j
\end{equation}
for each $p \in \cP$ \cite[Section 3.A, Equation (10)]{marshall2011inequalities}.

\paragraph{Majorization of quantum states}
Given two quantum states $\rho,\sigma\in \cD(\cH)$, we say $\sigma$ majorizes $\rho$, written $\rho\prec \sigma$ if $\vec\lambda(\rho) \prec \vec\lambda(\sigma)$. We say that $\varphi: \D(\cH) \to \R$ is \emph{Schur convex} if $\varphi(\rho)\leq \varphi(\sigma)$ for any $\rho,\sigma\in \cD(\cH)$ with $\rho \prec \sigma$. If $\varphi(\rho) < \varphi(\sigma)$ for any $\rho,\sigma\in \cD(\cH)$ such that  $\rho \prec \sigma$, and $\rho$ is not unitarily equivalent to $\sigma$, then $\varphi$ is \emph{strictly Schur convex}. We say $\varphi$ is Schur concave (resp.~strictly Schur concave) if $(-\varphi)$ is Schur convex (resp.~strictly Schur convex).

\section{Majorization flow} \label{sec:majflow}

Majorization, as defined in \eqref{def:majorize} above, is a pre-order, meaning it is reflexive ($p\prec p$) and transitive ($p\prec q \prec r \implies p \prec r$), and on the set of sorted probability vectors, $\cP^\downarrow = \left\{ p \in \R^d : p_1 \geq \dotsm \geq p_d \geq 0, \sum_{i=1}^d p_i = 1 \right\}$, it is a partial order, meaning it is also antisymmetric ($p\prec q$ and $q\prec p$ implies $p=q$). It also satisfies the so-called {\em{lattice}} property \cite{CV02}, meaning for any pair $p, q \in \cP^\downarrow$, there is a unique greatest lower bound $\inf(p,q)\in \cP^\downarrow$, which satisfies
\begin{itemize}
    \item $\inf(p,q) \prec p$ and $\inf(p,q) \prec q$
    \item if $r \in \cP^\downarrow$ is any other lower bound, meaning $r \prec p$ and $r\prec q$, then $r \prec \inf(p,q)$.
\end{itemize}
 Likewise, there is a unique least upper bound $\sup(p,q) \in \cP^\downarrow$ such that $p \prec \sup(p,q)$, $q \prec \sup(p,q)$, and  $\sup(p,q) \prec r$ for any $r\in \cP^\downarrow$ satisfying the relations $p \prec r$ and $q\prec r$.

In fact, majorization satisfies the stronger \emph{complete lattice} property, meaning the infimum and supremum of an arbitrary subset $S$ of $\cP^\downarrow$ exist and are unique, and can be obtained by an explicit algorithm \cite{YG19} (see also \cite{BBH+19}). For a possibly non-sorted set $S\subseteq \cP$, the supremum and infimum exist but are non-unique, since any permutation of a supremum (resp. infimum) is another supremum (resp. infimum).
We will be particularly interested in the case
\[
S = B_\eps(r) := \left\{ q \in \cP : \TV(q, r) \leq \eps \right\}
\]
for $\eps \geq 0$, where $\TV(q,r) := \frac{1}{2}\|q-r\|_1 = \frac{1}{2}\sum_{i=1}^d |q_i-r_i|$ is the total variation distance.
Although the vectors in $S$ may not be sorted, we can define $\inf B_\eps(r) := \pi\inv(\inf(B_\eps(r)^\downarrow))$ where for a set $S\subset \cP$,
\[
S^\downarrow := \left\{ p^\downarrow : p \in S \right\}
\]
where $p^\downarrow$ is a permutation of $p$ sorted in decreasing order, and $\pi$ is the permutation which sorts the elements of $r$ in decreasing order. This choice of ordering of the infimum minimizes the total variation distance between the infimum and the center $r$ of the ball $B_\eps(r)$.

The choice to use the total variation distance as the metric to define $B_\eps(r)$ is not arbitrary; in fact, the total variation ball admits the property that the majorization infimum $\inf(B_\eps(r))$ is attained on $B_\eps(r)$:
\[
\inf(B_\eps(r)) \in B_\eps(r).
\]

This is a very useful fact, which was first used by \cite{HY10} in the proof of their Theorem 2, in order to maximize the Shannon entropy over the total variation ball. This fact was independently rediscovered twice more, in \cite{HOS18} and by the present authors in \cite{HD17}, and was subsequently connected more directly to the notions of infima and suprema in \cite{BBH+19}. Henceforth, we write $\min(B_\eps(r))\equiv \inf(B_\eps(r))$ to emphasize the fact that the infimum is attained in the set itself.

In this work, we make a refined analysis which allows us to derive the path traced by the majorization minimizer as $\varepsilon$ is increased infinitesimally, and hence introduce the notion of \emph{majorization flow}. The latter notion is then used to provide simple proofs of uniform continuity bounds as well as novel Lipschitz continuity bounds for large classes of Schur concave functions, in both the classical and quantum setting (see \Cref{sec:q-to-c}).

\paragraph{The majorization flow}
Consider the map for $\eps \geq 0$
\[
\begin{aligned}
\mmm_\eps : \qquad \cP &\to \cP\\
 r &\mapsto \min(B_\eps(r)).
\end{aligned}
\]
We call $(\mmm_s)_{s\geq 0}$  the \emph{majorization flow}.
In \cite{HD17}, the present authors have shown that $(\mmm_s)_{s\geq 0}$ is a non-linear semigroup on $\cP$, in the sense that for $s,t \geq 0$,
\begin{equation}\label{eq:intro-semigroup}
\begin{aligned}
\mmm_0 &= \id\\
\mmm_{s+t} &= \mmm_s \circ \mmm_t.
\end{aligned}
\end{equation}
\cite{HOS18} established the additional property that $\mmm_s$ is \emph{majorization--preserving}, meaning that if $p \prec q$, then $\mmm_s(p) \prec \mmm_s(q)$, for any $s \geq 0$.
In this work, we show that this semigroup can be characterized by a (nonlinear) generator $\gen : \cP \to \R^d$ such that
\begin{equation}
\left.\partial_s^+ \mmm_s\right|_{s=0} = \gen
\end{equation}
where $\partial_s^+$ indicates the one-sided derivative from above. 
This leads to the following formulas.

\begin{theorem}\label{thm:integral-formulas}
Let $\eps > 0$ and $(\mmm_s)_{s\geq 0}$ be as defined above. Then
\begin{equation} \label{eq:mmm-integral-formula}
\mmm_\eps(r) = r + \int_0^\eps \gen(\mmm_s(r)) \d s.
\end{equation}
Additionally, for $H: \cP \to \R$ which is continuously differentiable on $\cP_+$, the quantity $\Gamma_H(r) := \partial_s^+ H(\mmm_s(r))|_{s=0}$ exists and satisfies
\begin{equation} \label{eq:H-integral-formula}
H(\mmm_\eps(r))  = H(r) + \int_0^\eps \Gamma_H(\mmm_s(r)) \d s.
\end{equation}
\end{theorem}
This result is proven in \Cref{sec:gen-and-proofs}.

\begin{remark}
The relations \eqref{eq:mmm-integral-formula} and \eqref{eq:H-integral-formula} are of particular importance. Equation \eqref{eq:mmm-integral-formula} allows one to determine properties of $\mmm_s(r)$ by simply analyzing $\gen$, while \eqref{eq:H-integral-formula} allows one to analyze continuity properties of $H$ using $\Gamma_H$.
\end{remark}

For $r\in \cP$, define the path $\gamma(t) = \mmm_t(r)$. We call the path $(\gamma(t))_{t=0}^1 \subseteq \cP$ as the \emph{path of majorization flow} starting from $r$, with respect to the total variation distance. Intuitively speaking, along this path the probability vector decreases in majorization order as quickly as possible, while changing at constant speed in $1$-norm (as we will see, $\frac{1}{2}\|\gen(r)\|_1=1$ for all $r\in \cP$). 

Majorization flow provides a simple and powerful tool for understanding how the majorization order changes with respect to total variation distance. Using this tool, we prove \Cref{thm:Delta-eps-Schur-convex} which provides a {\em{universal proof}} for many known entropic continuity bounds, and \Cref{thm:Convex-type-Lip} which establishes novel Lipschitz continuity bounds, including bounds on the $\alpha$-R\'enyi entropy (for $\alpha >1$) with an exponentially-improved dependence on $\alpha$ over previously known results.

\paragraph{Continuity bounds from majorization flow}

For $\eps >0$ and Schur concave $H$ (meaning $H(p) \geq H(q)$ if $q \prec p$), we have $H(\mmm_\eps(r)) = \max_{p \in B_\eps(r)} H(p)$ since $\mmm_\eps(r) = \min( B_\eps(r))$. In this case, using \eqref{eq:H-integral-formula} we obtain
\begin{equation} \label{eq:Delta-eps-by-Gamma}
\Delta_\eps^H(r):= \max_{p \in B_\eps(r)} H(p) - H(r) = \int_0^\eps \Gamma_H(\mmm_s(r)) \d s.
\end{equation}
In other words, the amount $H$ can locally increase near $r$ (quantified by $\Delta_\eps^H(r)$) is determined by $\Gamma_H$. Moreover, the global continuity properties of $H$ are determined by $\sup_{r\in \cP}\Delta_\eps^H(r)$. To see this, note that if $p,q \in \cP$ satisfy $\TV(p,q)\leq \eps$, then
\begin{equation}
|H(p) - H(q)| \leq \max\{\Delta_\eps(p), \Delta_\eps(q)\} \leq \sup_{r\in \cP} \Delta_\eps(r) \label{eq:uniform-bound-from-Delta-eps}
\end{equation}
This fact and \eqref{eq:Delta-eps-by-Gamma}  have two immediate consequences for the continuity properties of $H$:
\begin{enumerate}
	\item If $\Gamma_H$ is Schur convex, then by \eqref{eq:Delta-eps-by-Gamma}, $\Delta_\eps^H$ is Schur convex too, as $\mmm_s$ is majorization preserving for all $s\in [0,1]$. This provides the upper bound $\Delta_\eps^H(r) \leq \Delta_\eps(\psi)$ for $\psi = (1,0,\dotsc,0)$, (since $\psi$ majorizes every $r \prec \psi$ for every $r \in \cP$) which yields a tight uniform continuity bound by \eqref{eq:uniform-bound-from-Delta-eps}.
	\item If $\Gamma_H$ can be upper bounded by $k>0$ on $\cP_+$, then by \eqref{eq:Delta-eps-by-Gamma},  $\Delta_\eps^H(r)\leq \eps k$. This immediately yields a Lipschitz continuity bound for $H$ by \eqref{eq:uniform-bound-from-Delta-eps}.
\end{enumerate}

The second point can be rephrased as the following corollary to \Cref{thm:integral-formulas}.
\begin{corollary}\label{cor:opt-Lip}
Let $H: \cP\to \R$ be a Schur concave function which is continuously differentiable on $\cP_+$. We write $H(r_1,\dotsc,r_d) \equiv H(r)$ for $r\in \cP$. Next, for $r \in \cP$, let $i_+ \in \{1,\dotsc,d\}$ be an index such that $r_+ = r_{i_+}$, and similarly $i_- \in \{1,\dotsc,d\}$ such that $r_- = r_{i_-}$.
Define
\begin{equation}
\begin{aligned}
\Gamma_H : \quad \cP_+ &\to \R\\
 r &\mapsto (\partial_{r_{i_+}} - \partial_{r_{i_-}})H(r_1,\dotsc, r_d).
\end{aligned}
\end{equation}
Note that this definition does not depend on the choice of $i_\pm$ since $H$ is permutation invariant. Then $H$ is Lipschitz continuous if and only if
\[
k := \sup_{r \in \cP_+} \Gamma_H(r)
\]
satisfies $k < \infty$. Moreover, in the latter case $k$ is the optimal Lipschitz constant for $H$.
\end{corollary}
This is proven in \Cref{sec:gen-and-proofs}.

\paragraph{Comparison to other flows}
The notion of a flow arises naturally in various branches of physics, mathematics and engineering.  It is interesting to compare
and contrast the notion of majorization flow that we have introduced to the notion of gradient flow that arises in optimal transport and differential geometry, and has been applied to study open quantum systems.

The gradient flow induced by a function $F$ in a metric space can be loosely interpreted as the flow that decreases $F$ as quickly as possible \cite[p.~645]{Vil09}; in a similar sense, the majorization flow decreases the majorization order as quickly as possible (while following a unit speed path in total variation distance). However, there are several complications that prevent making the connection between majorization flow and gradient flow more precise:
\begin{itemize}
	\item Perhaps the most obvious one: decreasing in majorization order requires non-increasingness of all the partial sums given in \eqref{def:majorize}, while decreasing a function only requires decreasing a scalar value.
	\item The theory of gradient flow is well-developed on continuous spaces, such as probability measures on $\R^d$ equipped with a Wasserstein metrc (see, e.g. \cite{Vil09}), but much less well-developed in the discrete case considered here (probability measures on $\{1,\dotsc,d\}$).
	\item Here we consider the total variation distance, which can be seen as the $1$-Wasserstein distance induced by the Hamming distance on the set $\{1,\dotsc,d\}$. Almost all of the literature in discrete or continuous space takes the metric to be the $p$-Wasserstein distance for $p>1$ (for smoothness reasons).
\end{itemize}

\section{Quantum to classical reduction} \label{sec:q-to-c}
So far, we have only discussed majorization flow in the classical context of probability vectors $r \in \cP \subset \R^d$. Majorization of quantum states has also proven to be a useful concept (see e.g. \cite{Nielsen-LOCC}). In fact, every result in this article carries over to the framework of quantum states on a finite-dimensional Hilbert space $\cH$, by mapping quantum states $\rho \in \cD(\cH)$ to the corresponding probability vectors $\vec \lambda^\downarrow(\rho) \in \cP$ (consisting of eigenvalues of $\rho$ arranged in non-increasing order), due to \eqref{eq:unorm-q-to-c} below.

Let $A\in \Bsa(\cH)$. Then let $\Eig^\downarrow(A) = \diag (\vec \lambda^\downarrow(A))$ be the diagonal matrix with diagonal entries
\[
\lambda^\downarrow_1(A)\geq \lambda^\downarrow_2(A) \geq \dotsm \geq \lambda^\downarrow_d(A)
\]
are the sorted eigenvalues of $A$, counted with multiplicity, and we denote $\vec \lambda^\downarrow(A) := (\lambda^\downarrow_i(A))_{i=1}^d)$. By \cite[IV.62]{bhatia97}, we have for $A,B\in \Bsa(\cH)$
\begin{equation}
\unorm{\Eig^\downarrow(A)- \Eig^\downarrow(B)} \leq \unorm{A-B} \label{eq:unorm-q-to-c}
\end{equation}
for any unitarily invariant norm $\unorm{\cdot}$. Moreover, in the case of the trace distance,
\[
\frac{1}{2}\|\Eig^\downarrow(A)- \Eig^\downarrow(B)\|_1 = \frac{1}{2}\|\vec \lambda^\downarrow(A) - \vec \lambda^\downarrow(B)\|_1
\]
is the total variation distance between the sorted vectors of eigenvalues of $A$ and $B$.

This inequality justifies a natural definition of $\mmm_\eps$ for quantum states $\rho \in \cD(\cH)$: if $\rho = U \Eig^\downarrow(\rho) U^*$ for unitary $U$, then define
\[
\mmm_\eps(\rho) := U \diag(\mmm_\eps(\vec \lambda(\rho)))U^*.
\]
In other words, in the sorted eigenbasis of $\rho$, the majorization flow $(\mmm_t)_{0\leq t\leq \eps}$ simply updates the eigenvalues.
Then $\mmm_\eps(\rho) \in B_\eps(\rho)$ for all $\eps\geq 0$ by \eqref{eq:unorm-q-to-c}, and $\mmm_\eps(\rho)$ is the majorization minimizer of $B_\eps(\rho)$. H

Next, let $H: \cD(\cH)\subset \Bsa(\cH) \to \R$ be a unitarily invariant function. Then $H(A) = H(\Eig^\downarrow(A))$ is a function of the eigenvalues of $A$ alone: $H(A) = H_{\operatorname{cl}}(\vec \lambda^\downarrow(A))$ for some function $H_{\operatorname{cl}} : \cP \to \R$.

Hence, if $\frac{1}{2}\|\rho-\sigma\|\leq \eps$, then $\frac{1}{2}\|\vec \lambda^\downarrow(\rho) - \vec \lambda^\downarrow(\sigma)\|_1\leq \eps$ as well, and
\[
|H(\rho) - H(\sigma)| = |H_{\operatorname{cl}}(\vec \lambda^\downarrow(\rho)) - H_{\operatorname{cl}}(\vec \lambda^\downarrow(\sigma))|
\]
and it remains to bound this difference in terms of $\eps$, using $\frac{1}{2}\|\vec \lambda^\downarrow(\rho) - \vec \lambda^\downarrow(\sigma)\|_1\leq \eps$. This is precisely the task of establishing a classical continuity bound. Hence, the continuity bounds discussed in this article immediately extend to the quantum case.

\section{Generator of majorization flow and proof of  \texorpdfstring{\Cref{thm:integral-formulas} and Corollary~\ref{cor:opt-Lip}}{Theorem 3.1 and Corollary 3.2}} \label{sec:gen-and-proofs}

In \cite{HY10,HD17,HOS18}, it was shown that for any $r\in \cP$ and $\eps \geq 0$, there exists a probability vector $r_\eps^* \in \cP$ such that:
\begin{enumerate}
	\item $r_\eps^* \in B_\eps(r)$
	\item $r_\eps^* \prec q$ for any $q\in B_\eps(r)$.
\end{enumerate}
In \cite{HD17,HOS18}, a complementary probability vector $r_{*,\eps} \in B_\eps(r)$ was also constructed such that $q \prec r_{*,\eps}$ for any $q\in B_\eps(r)$. We will, however, only focus on $r^*_\eps$ in this article, however. The form of $r_\eps^*$ can be derived by maximizing the Shannon entropy over $B_\eps(r)$ using a waterfilling argument (as was done in \cite{HY10}), or as a special case of the form of maximizers of concave functions over $B_\eps(r)$ (as was shown in \cite{HD17}), or by constructing the majorization infimum of $B_\eps(r)$ (as suggested in \cite{BBH+19}). Quantum states corresponding to such probability vectors were also constructed in the more general thermal-majorization setting in~\cite{Remco,vNW17}.

The explicit construction of $r_\eps^*$ was given in \cite{HY10,HD17,HOS18}, and is somewhat involved. However, the semigroup property \eqref{eq:intro-semigroup} of $\mmm_\eps$ shows that only the generator $\cL$ of majorization flow is needed to construct and analyze the behavior of $r^*_\eps$ for any $\eps \in [0,1]$. Hence, in this article, instead of recalling the explicit form of $r_\eps^*$ for all $\eps \in [0,1]$, we simply use the semigroup property \eqref{eq:intro-semigroup}, and the form of $r_\eps^*$ for small $\eps$ (small compared to a certain quantity $\delta(r)$ defined below).

First, setting $\mmm_\eps(r) = r_\eps^*$, we have the crucial semigroup property mentioned in \eqref{eq:intro-semigroup} that
\begin{equation*} 
\mmm_{\eps_1 + \eps_2} = \mmm_{\eps_1}\circ \mmm_{\eps_2}
\end{equation*}
for all $\eps_1,\eps_2 \in [0,1]$, where we define $\mmm_{\eps}(r) = u$ for $\eps > 1$, where $u = (\frac{1}{d},\dotsc, \frac{1}{d})$ is the uniform distribution. This property was established in \cite{HD17}.

The second property is as follows: let $r\in \cP$. Then $r_\eps^*$ is a linear perturbation of $r$ for $\eps$ small enough, in the sense that
\begin{equation} \label{eq:linear-perturbation}
r_\eps^* = r + \eps \gen(r) \qquad  \forall \eps \leq \delta(r)
\end{equation}
where $\gen(r) \in \R^d$ and $\delta(r) \geq 0$ are defined as follows. If $r = u$ is the uniform distribution, set $\cL(r) = 0$. For $r\neq u$, denote
\begin{itemize}
	\item the largest entry of $r$ as $r_+$, and its multiplicity by $k_+$
	\item the second largest (distinct) entry of $r$ as $\mu_+$
	\item the smallest entry of $r$ as $r_-$, and its multiplicity by $k_-$
	\item and the second smallest (distinct) entry of $r$ as $\mu_-$.
\end{itemize}
Then, $\gen(r) \in \R^n$ is a vector defined by
\begin{equation} \label{eq:def_X}
\gen(r)_i = \begin{cases}
-\frac{1}{k_+} & r_i = r_+\\
\frac{1}{k_-} & r_i = r_-\\
0 & \text{else,}
\end{cases}
\end{equation}
for $i=1,\dotsc, d$, and $\delta(r)$ is defined as
\[
\delta(r) := \begin{cases}
0 &r = u\\
\min \left\{ k_+ (r_+ - \mu_+), k_- (\mu_- - r_-) \right\} & \text{else.}
\end{cases}
\]
Note that $\delta(r) > 0$ if $r\neq u$.
With these definitions, equation~\eqref{eq:linear-perturbation} follows directly from the constructions presented in \cite{HY10,HD17,HOS18}.
Together, \eqref{eq:intro-semigroup} and \eqref{eq:linear-perturbation} provide all the information needed to construct $r_\eps^*$ for any $\eps \in [0,1]$; moreover, $\eps \mapsto r_\eps^*$ is piecewise affine, with at most $d$ pieces.

Now, we may turn to the proof of \Cref{thm:integral-formulas}.
\paragraph{Proof of \Cref{thm:integral-formulas}}

To establish \eqref{eq:mmm-integral-formula}, note that equation \eqref{eq:linear-perturbation} immediately yields
\[
\left.\partial_\eps^+ \mmm_\eps(r)\right|_{\eps=0} = \gen(r)
\]
where $\partial_\eps^+$ indicates the one-sided derivative in $\eps$ from above. Then
\[
\left.\partial s^+ \mmm_s(r)\right|_{s=t} = \left.\partial_t^+ \mmm_{s+t}(r)\right|_{t= 0}= \left.\partial_t^+ \mmm_{t}(\mmm_s(r))\right|_{t= 0} = \gen(\mmm_s(r))
\]
follows using the semigroup property. As the path $s\mapsto \mmm_s(r)$ is piecewise affine with at most $d$ pieces, for each fixed $r\in \cP$, the two-sided derivative $\partial_s \mmm_s(r)$ exists for all but at most $d$ elements $s\in [0,1]$. Then, using the fundamental theorem of calculus,
\[
\mmm_\eps(r) = r + \int_0^\eps \gen(\mmm_s(r)) \d s.
\]
This establishes \eqref{eq:mmm-integral-formula}. \Cref{eq:H-integral-formula} follows in the same manner by considering $H \circ \mmm_\eps(r)$ instead of $\mmm_\eps(r)$, and using that $H$ is continuously differentiable on $\cP_+$. \qed

\bigskip

Since $\frac{1}{2}\|\gen(r)\|_1 = 1$ for all $r\in \cP\setminus\{u\}$ as shown by \eqref{eq:def_X}, and $\gen(u) = 0$, the triangle inequality 
\[
\left\|\int_0^\eps \gen(\mmm_s(r)) \d s \right\|_1 \leq \int_0^\eps \|\gen(\mmm_s(r))\|_1 \d s 
\]
immediately yields
\[
\frac{1}{2}\|\mmm_\eps(r) - r\|_1 \leq \eps.
\]

We summarize the properties of $\mmm_\eps$ in the following proposition.

\begin{proposition}[Properties of $\mmm_\eps$] \label{prop:properties_of_Lambda_eps} Let $\sigma\in \D(\cH)$. We have the following properties of $\mmm_\eps$, for any $\eps\in (0,1]$.
\begin{enumerate}[label*=\alph*.,ref=(\alph*)]
	\item \label{item:Lambda_epsstates-to-states} Maps probability vectors to probability vectors: $\mmm_\eps : \cP \to \cP$. 
	\item \label{item:Lambda_eps_min_maj_order} Minimal in majorization order:   $\mmm_\eps(r)\in \Be(r)$ and for any $p\in \Be(r)$, we have $\mmm_\eps(r) \prec p$.
\item  \label{item:Lambda_eps-semigroup} Semigroup property:
if $\eps_1,\eps_2\in (0,1]$ with $\eps_1+\eps_2 \leq 1$, we have $\mmm_{\eps_1+\eps_2}(\sigma) = \mmm_{\eps_1} \circ \mmm_{\eps_2}(\sigma)$.
\item \label{item:Lambda_eps-maj-preserving} Majorization-preserving: let $p,q\in \cP$ such that $p \prec q$. Then $\mmm_\eps(p) \prec \mmm_\eps(q)$.
\item \label{item:Lambda_eps_fixed_point} $u = (1/d,\dotsc,1/d)\in\cP$ is the unique fixed point of $\mmm_\eps$, i.e.\@~the unique solution to $r = \mmm_\eps(r)$ for $r \in \cP$.
\item \label{item:Lambda-eps-near-tau}For any state $r \in \Be(u)$, we have $\mmm_\eps(r) = u$.
\item \label{item:Lambda-eps-on-psi}For any extremal probability vector $q$, i.e.  $q=\pi(1,0,\dotsc,0)$ for some permutation $\pi$, the probability vector $\mmm_\eps(q)$ has the form
\begin{equation}
\mmm_\eps(q) = \begin{cases} 
\pi(1- \eps, \frac{\eps}{d-1},\dotsc \frac{\eps}{d-1}) & \eps < 1 - \frac{1}{d}\\
u = (1/d,\dotsc,1/d) & \eps \geq 1 - \frac{1}{d}.
\end{cases} \label{eq:Lambda-eps-pure}
\end{equation}
\end{enumerate}
 
\end{proposition}
\begin{proof}	

The proof of properties \ref{item:Lambda_epsstates-to-states} and \ref{item:Lambda_eps_min_maj_order} can be found in \cite{HD17,HOS18}; the property \ref{item:Lambda_eps-semigroup} was proved in in \cite{HD17}, property \ref{item:Lambda_eps-maj-preserving} can be found in Lemma 2 of \cite{HOS18}. The property \ref{item:Lambda_eps_fixed_point} is immediate from \eqref{eq:linear-perturbation}.  Lastly, the properties \ref{item:Lambda-eps-near-tau} follows from the majorization minimizer property, and \ref{item:Lambda-eps-on-psi} follows from \eqref{eq:linear-perturbation}.
\end{proof}

\paragraph{Calculation of $\Gamma_H$ for symmetric $H$}
Letting $r_\pm$ be the largest and smallest elements of $r$, note that
\begin{align*}	
\Gamma_H(r) &=  \left.\frac{\d}{\d y} H(\mmm_y(r))\right|_{y=0}\\
&= \sum_{i=1}^d \gen(r)_i H_i(r)= \sum_{i : r_i = r_+} \frac{-1}{k_+} H_i(r) + \sum_{i : r_i = r_-} \frac{1}{k_-} H_i(r)
\end{align*}
where $H_i(r) = \partial_{r_i}H(r)$. Assume $H$ is symmetric\footnote{Note that symmetry is implied by Schur concavity.} (invariant under permutations of $r$). Then $H_i(r) = H_j(r)$ if $r_i = r_j$. In that case,
\begin{equation} \label{eq:Gamma-for-symmetric-H}
\Gamma_H(r) = H_-(r) - H_+(r)
\end{equation}
where $H_\pm(r) = H_i(r)$ for any $i$ such that $r_i= r_\pm$. That is, $\Gamma_H(r)$ is simply the difference between two partial derivatives of $H$, evaluated at $r$. 

Moreover, if $H$ is Schur concave, then \eqref{eq:Gamma-for-symmetric-H} shows that 
\[
 \Gamma_H(r) \geq 0 \qquad \forall r \in \cP
\]
using \eqref{eq:S-convex-condition}, and likewise if $H$ is Schur convex, then $\Gamma_H(r) \leq 0$ for all $r \in \cP$.

With \eqref{eq:Gamma-for-symmetric-H} established, we now prove \Cref{cor:opt-Lip}.

\begin{proof}[Proof of \Cref{cor:opt-Lip}]
In the case that $\sup_{r\in \cP_+}\Gamma_H(r) = \infty$, for each $n \in \mathbb{N}$, there exists $r^{(n)} \in \cP_+$ such that
\[
 \lim_{s\downarrow 0}\frac{H(\mmm_s(r^{(n)})) - H(r^{(n)})}{s} > n.
\]
Additionally, $r^{(n)}\neq u$, otherwise the numerator would be zero. Hence, for any $\eps > 0$ there exists $0 < s_n < \TV(r^{(n)}, u)$ such that
\[
\frac{H(\mmm_{s_n}(r^{(n)})) - H(r^{(n)})}{s_n} > n - \eps.
\]
Since $s_n < \TV(r^{(n)}, u)$, we have $s_n = \TV(\mmm_{s_n}(r^{(n)}), r^{(n)} )$, and hence
\[
\frac{H(\mmm_{s_n}(r^{(n)})) - H(r^{(n)})}{ \TV(\mmm_{s_n}(r^{(n)}), r^{(n)} )} > n - \eps.
\]
Any Lipschitz constant $k$ must be larger than the left-hand side, for any $n$, and hence must be $\infty$.

Next, consider the case in which $\sup_{r\in \cP_+}\Gamma_H(r)$ is finite. \Cref{eq:Gamma-for-symmetric-H} shows that the quantity defined in \Cref{cor:opt-Lip} is the $\Gamma_H$ of \Cref{thm:integral-formulas}. Hence, in the case that $\sup_{r\in \cP_+}\Gamma_H(r) < \infty$ holds, \eqref{eq:Delta-eps-by-Gamma} and \eqref{eq:uniform-bound-from-Delta-eps} show that $\sup_{r\in \cP_+}\Gamma_H(r)$ is indeed a Lipschitz constant for $H$. It remains to show this constant is optimal.

Assume there is some Lipschitz constant $k' < \sup_{r\in \cP_+}\Gamma_H(r)$.
For each $\eps > 0$, let $r^{(\eps)} \in \cP_+$ satisfy $\Gamma_H(r^{(\eps)}) > \sup_{r\in \cP_+}\Gamma_H(r) - \eps$. Then
\[
k' \geq \frac{H(\mmm_s(r^{(\eps)})) - H(r^{(\eps)})}{\frac{1}{2}\|\mmm_s(r^{(\eps)}) -r^{(\eps)}\|_1 } \geq \frac{H(\mmm_s(r^{(\eps)})) - H(r^{(\eps)})}{s}
\]
for each $s \in (0,1]$. Taking the limit $s\to 0$ yields
\[
k' \geq \Gamma_H(r^{(\eps)}) > \sup_{r\in \cP_+}\Gamma_H(r) - \eps.
\]
Taking $\eps \to 0$ shows $k' \geq \sup_{r\in \cP_+}\Gamma_H(r)$, which contradicts $k' < \sup_{r\in \cP_+}\Gamma_H(r)$. Hence, the quantity $\sup_{r\in \cP_+}\Gamma_H(r)$ is indeed the optimal Lipschitz constant.
\end{proof}

\section{Establishing entropic continuity bounds via majorization flow}\label{sec:cty-from-majflow}

There are several families of single-partite entropies: $\alpha$-R\'enyi entropies, Tsallis entropies, unified entropies, and so forth (definitions and references for these can be found in \Cref{sec:tsallis-renyi,sec:more-cty-bounds}). These various entropies have a lot of similarities, in particular sharing a functional form: for $r \in \cP_d$,
\begin{equation}\label{eq:functional-form-h-phi}
H(r) = h\left( \sum_{i=1}^d \phi(r_i)\right)
\end{equation}
for a pair of scalar functions $h$ and $\phi$. However, they exhibit different mathematical properties, most notably concavity or lack thereof. For example, the $\alpha$-R\'enyi entropies, defined by
\[
 H_\alpha(r)= \frac{1}{1-\alpha}\log \left( \sum_{i=1}^d r_i^\alpha \right)
 \] are concave for $\alpha \in (0,1)$, but are neither concave nor convex for $\alpha > 1$. Certain properties of the latter have proven particularly difficult to study, perhaps as a consequence of the lack of concavity.

While Audenaert proved a tight uniform continuity bound for the $\alpha$-R\'enyi entropies for $\alpha \in (0,1)$ in 2007 (\cite{Audenaert07}; see \eqref{eq:uniform_Renyi_bound} below), a uniform continuity bound on the $\alpha$-R\'enyi entropies for $\alpha > 1$ was not established until 2011. At that time, Rastegin \cite{Ras2011} proved the bound
\begin{equation}
|H_\alpha(p) - H_\alpha(q)| \leq  \frac{d^{2(\alpha-1)}}{\alpha-1}[  1  - (1-\eps)^\alpha  -\eps^\alpha (d-1)^{1-\alpha} ] \label{eq:bound-Rastegin}
\end{equation}
where $\eps = \TV(p,q)$, for $\alpha > 1$.
This bound, however, suffers from an exponential dependence on $\alpha$ (and for fixed $\alpha$, polynomial dependence on $d$), while scaling linearly with $\eps$, as $1  - (1 - \eps)^\alpha \approx \alpha \eps - \frac{1}{2}\alpha(\alpha-1) \eps^2 + O(\eps^3)$. Since the inequality $|H_\alpha(p) - H_\alpha(q)| \leq 2 \log d$ holds trivially, for even moderately large $\alpha$, \eqref{eq:bound-Rastegin} provides a non-trivial bound for a very small range of $\eps$.

In 2017, Chen et al (\cite{Renyi-CMNF}) improved upon this bound, showing that for $\alpha > 1$,
\begin{equation} \label{eq:Chen-renyi-bound}
|H_\alpha(p) - H_\alpha(q)| \leq \frac{d^{\alpha-1}}{\alpha-1}[  1  - (1-\eps)^\alpha  -\eps^\alpha (d-1)^{1-\alpha} ] 
\end{equation}
However, this bound still suffers from exponential dependence on $\alpha$.
The proof of both bounds proceeds by reducing to the case of the Tsallis entropy,
\[
T_\alpha(p) = \frac{1}{1-\alpha}\left(\Big(\sum_{i=1}^d p^\alpha\Big) - 1\right),
 \]
picking up an exponential prefactor along the way.

In this article, we prove that the dimensional dependence is at most linear\footnote{This is also established in-effect by \cite[Theorem 7, (2)]{WH19}, which was developed independently and posted slightly later than the first preprint of this article in which we established the linear bound discussed here.}, and in fact
\begin{equation} \label{eq:intro-Renyi-Lip}
|H_\alpha(p) - H_\alpha(q)| \leq  \frac{d\alpha}{\alpha-1} \frac{1}{\ln(2)} \eps
\end{equation}
for $\TV(p,q) \leq \eps$ and $\alpha > 1$. In fact, we prove that the Tsallis entropy increases the fastest near the corners of the probability simplex (i.e. the extremal points which are permutations of $(1,0,\dotsc,0)$) and the slowest near the center of the simplex (see \Cref{thm:Delta-eps-Schur-convex}), while the $\alpha$-R\'enyi entropy increases the fastest close to the center of the probability simplex. This mismatch shows why bounding the difference of R\'enyi entropies of two probability vectors by the difference of Tsallis entropies of the same two distributions does not work well: a large prefactor is needed to bound the rapidly-changing R\'enyi entropies near the center of the simplex by the Tsallis entropies which change the slowest there. With the benefit of hindsight (and our proof techniques), we can find that indeed, a linear prefactor suffices to compare the maximum differences in R\'enyi entropies between two probability vectors which are at a fixed total variation distance apart, and the maximum difference in Tsallis entropies between two distributions at the same distance apart; however, these two maximum differences occur at very different parts of the probability simplex. These two effects can be seen quantitatively in \Cref{prop:compare-Renyi-Tsallis}.

To prove the bound \eqref{eq:intro-Renyi-Lip}, as well as determine where each entropy increases the fastest, we take a unified approach to establishing  entropic continuity bounds. While concavity only holds for certain entropies, we exploit the fact that all the above entropies are \emph{Schur concave}, meaning they are decreasing in the majorization order: $H(q) \leq H(p)$ if $p \prec q$.  Majorization thus provides a tool for understanding properties of a wide class of entropies in a unified manner.

We consider a class of entropic functionals called $(h,\phi)$-entropies, which were introduced by \cite{SMMP93a}. These are defined by the formula
\[
 H_{(h,\phi)}(r) := h\left(\sum_{i=1}^d \phi(r_i)\right)
\]
for $r \in \cP$, and likewise by $H_{(h,\phi)}(\rho) = h(\tr(\phi(\rho)))$ for $\rho \in \cD(\cH)$, using the functional calculus. The concept of $(h,\phi)$-entropies for quantum states was introduced by \cite{Bosyk2016}. In other words, the $(h,\phi)$-entropy of a quantum state is defined by the $(h,\phi)$-entropy of the probability vector given by its eigenvalues.

We will consider two classes of $(h,\phi)$-entropies which capture almost all single-partite entropies considered in the literature.
\begin{itemize}
	\item We say a function $S: \cP\to \R$ is a  \emph{\typeone{} $(h,\phi)$-entropy} if $S = H_{(h,\phi)}$ for some  $\phi: [0,1] \to \R$ which is continuously differentiable on $(0,1]$ and $h: [\phi(1), \phi(\tfrac{1}{d})d] \subset \R \to \R$ which is continuously differentiable on $(\phi(1), \phi(\tfrac{1}{d})d]$, such that $\phi(0) = 0$ and $h(\phi(1)) = 0$, with $h$  strictly increasing and (not necessarily strictly) concave, and $\phi$  strictly concave.
	\item We say a function $S: \cP\to \R$ is a  \emph{\typetwo{}  $(h,\phi)$-entropy} if $S = H_{(h,\phi)}$ for some  $\phi: [0,1] \to \R$ which is continuously differentiable on $(0,1]$ and $h:  [\phi(\tfrac{1}{d})d, \phi(1)] \subset \R \to \R$ which is continuously differentiable on $ [\phi(\tfrac{1}{d})d, \phi(1))$, such that $\phi(0) = 0$ and $h(\phi(1)) = 0$, with $h$  strictly decreasing and (not necessarily strictly) convex, and $\phi$  strictly convex.
\end{itemize}
Both classes of $(h,\phi)$-entropies are strictly Schur concave; this follows immediately from the fact that symmetric and strictly convex (resp.~strictly concave) functions are strictly Schur convex (resp.~strictly Schur concave), and that the composition with a strictly increasing function preserves strict Schur convexity and strict Schur concavity. In contrast, composition with a strictly decreasing function swaps strict Schur convexity and strict Schur concavity. Likewise, both classes of $(h,\phi)$-entropies are continuously differentiable on $\cP_+$.

\paragraph{Examples of $(h,\phi)$-entropies}
\begin{itemize}
    \item The $\alpha$-R\'enyi entropy for $\alpha <1$, the $\alpha$-Tsallis entropy for $\alpha>0$, the von Neumann entropy (or Shannon entropy in the classical case), the $(s,\alpha)$-unified entropies for $\alpha\in(0,1)$ and $s\leq 1$, and entropies induced by $f$ divergences with strictly convex $f$ are all \typeone{} $(h,\phi)$-entropies.
    \item The $\alpha$-R\'enyi entropy for $\alpha >1$ and the $(s,\alpha)$-unified entropies for $\alpha > 1$ and $s\leq 1$ are  \typetwo{} $(h,\phi)$-entropies.
\end{itemize}
The R\'enyi and Tsallis entropies are discussed in more detail in \Cref{sec:tsallis-renyi}, and the other entropies are discussed in \Cref{sec:more-cty-bounds}.

While both \typeone{} and \typetwo{} $(h,\phi)$-entropies are Schur concave, \typeone{} $(h,\phi)$-entropies are additionally concave, as the composition of a concave increasing function with a concave function. On the other hand, in general \typetwo{} $(h,\phi)$-entropies are neither convex nor concave. In this article, we investigate the continuity properties of these two classes of entropies.

Let $H=H_{(h,\phi)}$ be an $(h,\phi)$-entropy (of either type). By \eqref{eq:Gamma-for-symmetric-H}, we have immediately that
\begin{equation}
\Gamma_H(r) = \left. \frac{\d}{\d t}  H_{(h,\phi)}(\mmm_{t}(r))\right|_{t=0} = h'( {\textstyle \sum_i} \phi(r_i) ) (  \phi'(r_-) - \phi'(r_+) ) \label{eq:Gamma_hphi}
\end{equation}
where $r_+$ denotes the largest element of $r$, and $r_-$ the smallest. Note that $r \mapsto r_-$ is Schur concave, while $r \mapsto r_+$ is Schur convex. 

\paragraph{\typeone{} $(h,\phi)$-entropies}
We see that for a \typeone{} $(h,\phi)$-entropy $H \equiv H_{(h,\phi)}$,
\[
r \mapsto (  \phi'(r_-) - \phi'(r_+) )
\]
is Schur convex and strictly positive for $r$ non-uniform, and likewise
\[
r \mapsto h'( {\textstyle \sum_i} \phi(r_i) )
\]
is strictly positive and strictly Schur convex. Hence, $\Gamma_H$ is strictly Schur convex on $\cP_+$. Then, following the discussion at the start of \Cref{sec:cty-from-majflow}, we obtain the following results.

\begin{theorem} \label{thm:Delta-eps-Schur-convex}
Let $H_{(h,\phi)}$ be a \typeone{} $(h,\phi)$-entropy. Then $\Gamma_H$ is strictly Schur convex on $\cP_+$, and
\[
r \mapsto \sup_{p \in B_\eps(r)} H_{(h,\phi)}(p) - H_{(h,\phi)}(r)
\]
is strictly Schur convex on $\cP$.
\end{theorem}
\begin{corollary}[Tight uniform continuity bounds for \typeone{} $(h,\phi)$-entropies] \label{cor:tight-uniform-bounds-concave-type}
For $\eps \in (0,1]$ and any states $p,q\in \cP$  such that $\TV(p,q)\leq \eps$, we have
 \begin{equation} 
 | H_{(h,\phi)}(p) - H_{(h,\phi)}(q) | \leq g(\eps) \label{eq:hphi-uniform-bound} 
 \end{equation}
 where
 \begin{equation} \label{eq:def_g}
 g(\eps) := \begin{cases}
h( \phi(1-\eps) + (d-1) \phi( \frac{\eps}{d-1})) & \eps < 1-\frac{1}{d}\\
h(d\phi(\frac{1}{d})) & \eps \geq 1 - \frac{1}{d}
\end{cases} 
 \end{equation}
and $d$ is the dimension. Moreover, equality in~\eqref{eq:hphi-uniform-bound} occurs if and only if one of the two distributions (say, $q$) is extremal (i.e.\@~a permutation of $(1,0,\dotsc,0)$), and either
\begin{enumerate}
	\item $\eps < 1 - \frac{1}{d}$ and $p = \pi( 1- \eps, \frac{\eps}{d-1},\dotsc, \dotsc, \frac{\eps}{d-1})$ for some permutation $\pi$, or
	\item  $\eps \geq 1- \frac{1}{d}$, and $p = u$ is uniform.
\end{enumerate}\label{thm:hphi-GCB}

\end{corollary}
This provides a tight uniform continuity bound for the Tsallis entropies, the $\alpha$-R\'enyi entropies for $\alpha \in (0,1)$, the Shannon entropy, the $(s,\alpha)$-unified entropies with $\alpha\in (0,1)$ and $s\leq 1$, and any entropy induced by an $f$-divergence or maximal $f$-divergence with strictly convex $f$. See \Cref{sec:more-cty-bounds} for more details and references.

Given an $(h,\phi)$-entropy, we may also consider its smoothed variant,
\begin{equation} \label{eq:def-H-smoothed}
H_{(h,\phi)}^\delta(p) := \max_{q \in B_\delta(p)} H_{(h,\phi)}(q) = H_{(h,\phi)}\circ \mmm_\delta(p)
\end{equation}
for $\delta\in[0,1]$. If $H_{(h,\phi)}$ is \typeone{}, we can simply establish Lipschitz continuity bounds for any $\delta > 0$ by using the Schur concavity of $\Gamma_{H_{(h,\phi)}}$ and \Cref{cor:opt-Lip}.


\begin{corollary} \label{prop:typ1-Lip}
Let $H_{(h,\phi)}^\delta$ be the smoothed variant of a \typeone{} $(h,\phi)$-entropy, for $\delta \in [0, 1]$ (as defined in \eqref{eq:def-H-smoothed}). Then $H_{(h,\phi)}^\delta$ is Lipschitz continuous on $\cP$ if and only if
\begin{equation}\label{eq:opt-Lip-smoothed}
k := \lim_{\eps \to 0} \frac{g(\eps+\delta)-g(\delta)}{\eps} < \infty
\end{equation}
where $g(\eps) := H_{(h,\phi)}(\psi_\eps^*)$ for $\psi_\eps^* = \diag(1-\eps, \frac{\eps}{d-1},\dotsc, \frac{\eps}{d-1})$ is given in \eqref{eq:def_g}. Moreover, if $k$ is finite, then it is the optimal Lipschitz constant for $H_{(h,\phi)}^\delta$. In particular, if $\delta >0$, then
\[
k = g'(\delta) = h'(\phi(1-\delta) + (d-1) \phi(\tfrac{\delta}{d-1}))(\phi'(\tfrac{\delta}{d-1}) - \phi'(1-\delta))
\]
and $H_{(h,\phi)}^\delta$ is Lipschitz continuous.
\end{corollary}
\begin{proof}
If $\delta = 0$, the result follows from \Cref{cor:tight-uniform-bounds-concave-type}. Hence, consider the case $\delta > 0$. For $r\in \cP_+$,
\begin{align*}	
\Gamma_{H_{(h,\phi)}^\delta}(r) &= \lim_{s\downarrow 0} \frac{1}{s}\left(H_{(h,\phi)}^\delta(\mmm_s(r)) - H_{(h,\phi)}^\delta(r)\right)\\
&= \lim_{s\downarrow 0} \frac{1}{s}\left(H_{(h,\phi)}(\mmm_{s+\delta}(r)) - H_{(h,\phi)}(\mmm_\delta(r))\right)\\
&= \lim_{s\downarrow 0} \frac{1}{s}\left(H_{(h,\phi)}(\mmm_{s}(\mmm_\delta(r)) - H_{(h,\phi)}(\mmm_\delta(r))\right)\\
&= \Gamma_{H_{(h,\phi)}}(\mmm_\delta(r)).
\end{align*}
Since $\Gamma_{H_{(h,\phi)}}$ is strictly Schur convex on $\cP_+$ by \Cref{thm:Delta-eps-Schur-convex}, and for any $\delta > 0$, the map $\mmm_\delta: \cP\to \cP_+$ is majorization-preserving, we have that $r \mapsto \Gamma_{H_{(h,\phi)}}(\mmm_\delta(r))$ is strictly Schur convex on $\cP$. Hence, it is maximized at $r = \psi := (1,0,\dotsc,0)$ (or any permutation thereof). Invoking \Cref{cor:opt-Lip} completes the proof.
\end{proof}

\paragraph{\typetwo{} $(h,\phi)$-entropies}

For a \typetwo{} $(h,\phi)$-entropy $H \equiv H_{(h,\phi)}$, 
\[
r \mapsto (  \phi'(r_+) - \phi'(r_-) )
\]
is Schur convex and strictly positive for non-uniform $r$, while
\[
r \mapsto -h'( {\textstyle \sum_i} \phi(r_i) )
\]
is strictly Schur concave and strictly positive. Hence, $\Gamma_H$ is the product of a Schur convex and Schur concave function. The former only depends on the largest and smallest entries of $r$, however. In the case $d=2$, these are all the entries, and for $x\in (0, \frac{1}{2}]$,
\[
\Gamma_H(\left\{ x, 1-x \right\}) = h'(\phi(x) + \phi(1-x)) (\phi'(x) - \phi'(1-x)) = \bin_{(h,\phi)}'(x)
\]
is the derivative of the binary $(h,\phi)$-entropy, where $\bin_{(h,\phi)}(x) := H_{(h,\phi)}(\left\{ x, 1-x \right\})$.
In the case $d > 2$, define
\[
\bar r = (r_1, \underbrace{z, \dotsc, z}_{d-2 \text{ times}}, r_d)
\]
where without loss of generality, $r_1 \geq r_2 \geq \dotsm \geq r_d$ are the sorted elements of $r$. Then $\Gamma_H(r) \leq \Gamma_H(\bar r)$, as $\bar r \prec r$ while having the same largest and smallest elements.

\begin{theorem} \label{thm:Convex-type-Lip}
Let $H_{(h,\phi)}$ be a \typetwo{} $(h,\phi)$-entropy.
If $d = 2$, then
\begin{equation}\label{eq:Convex-type-Lip-d-2}
\sup_{r\in \cP_+}\Gamma_H(r) = \sup_{0 < x \leq \frac{1}{2}} \bin'_{(h,\phi)}(x).
\end{equation}
If $d > 2$, then
\begin{equation} \label{eq:Convex-type-Lip-d-gt-2}
\sup_{r\in \cP_+}\Gamma_H(r) =\sup_{\substack{x, y: \\ 0 < x\leq \frac{1}{d} \leq y\\x \leq z \leq y}} -h'(\phi(y) + (d-2)\phi(z) + \phi(x))(\phi'(y) - \phi'(x)), \qquad z := \frac{1-y-x}{d-2}
\end{equation}
In either case, if $\phi$ is differentiable at zero and $h$ is differentiable at $\phi(1)$, then $H_{(h,\phi)}$ is Lipschitz continuous, and the optimal Lipschitz constant is given by $\sup_{r\in \cP_+}\Gamma_H(r)$.
\end{theorem}
\begin{proof}	
The discussion before \Cref{thm:Convex-type-Lip} and \eqref{eq:Gamma_hphi} establishes the expressions for $\sup_{r\in \cP_+}\Gamma_H(r)$, and the proof concludes by \Cref{cor:opt-Lip}.
\end{proof}
\begin{remark}
While \eqref{eq:Convex-type-Lip-d-2} (resp. \eqref{eq:Convex-type-Lip-d-gt-2}) do not provide a closed-form expression for $\sup_{r \in \cP_+} \Gamma_H(r)$, they reduce the naively $d$-dimensional optimization problem to a 1- (resp.~2-) dimensional problem.
\end{remark}

\subsection{Continuity bounds for the R\'enyi and Tsallis entropies} \label{sec:tsallis-renyi}
\paragraph{R\'enyi entropies}

The $\alpha$-R\'enyi entropy \cite{Ren61}  for $\alpha \in (0,1)\cup(1,\infty)$, of a quantum state $\rho \in \cD(\cH)$ is defined by
\[
H_\alpha(\rho) := \frac{1}{1-\alpha} \log \left(\tr \rho^\alpha\right).
\]
$H_\alpha$ is the $(h,\phi)$-entropy with $h(x) = \frac{1}{1-\alpha}\log x$ for $x\in \RR$ and $\phi(x) = x^\alpha$ for $x\in [0,1]$. For $\alpha\in (0,1)$, $h$ is concave and strictly increasing and $\phi$ is strictly concave. For $\alpha >1$, $h$ is convex and strictly decreasing, and $\phi$ is strictly convex. Hence, $H_\alpha$ is a \typeone{} $(h,\phi)$-entropy for $\alpha \in (0,1)$, and is a \typetwo{} $(h,\phi)$-entropy for $\alpha > 1$. It is known that $\lim_{\alpha \to 1}H_\alpha(\rho) = H(\rho)$, and $\lim_{\alpha\to\infty}H_\alpha(\rho) = H_\infty(\rho) := - \log \lambda_{\max}(\rho)$, where $\lambda_{\max}(\rho)$ denotes the largest eigenvalue of $\rho$.

For a probability distribution $p \in \cP$, the above quantity reduces to
\[
H_\alpha(p) = \frac{1}{1-\alpha} \log \left(\sum_{i=1}^d p_i^\alpha\right),
\]
and $H_\infty(p) = - \log \max_{1\leq i \leq d} p_i$. 

\paragraph{Tsallis entropies}

The $\alpha$-Tsallis entropy \cite{Tsallis1988} for $\alpha\in(0,1)\cup(1,\infty)$ of a quantum state $\rho \in \cD(\cH)$ is defined by
\[
T_\alpha(\rho) := \frac{1}{1-\alpha}[\tr(\rho^\alpha)-1],
\]
and in the case of a probability vector $p \in \cP$,
\[
T_\alpha(p) = \frac{1}{1-\alpha}\left[\sum_{i=1}^d p_i^\alpha - 1\right].
\]
The Tsallis entropy can be seen as a version of the R\'enyi entropy in which the logarithm has been linearized (up to a factor of $\ln 2$), using the first-order Taylor series $\log x \approx \frac{1}{\ln 2} (x - 1)$. The Tsallis entropy is not additive under tensor products (it is \emph{nonextensive}) and instead satisfies the relation
\begin{equation}\label{eq:Talpha_non_subadd}
T_\alpha(\rho\otimes \sigma) = T_\alpha(\rho) + T_\alpha(\sigma) + (1-\alpha)T_\alpha(\rho)T_\alpha(\sigma)
\end{equation}
for $\alpha \in (0,1)\cup(1,\infty)$ and $\rho,\sigma\in \cD(\cH)$, as can be verified by direct computation.

We have that $T_\alpha = H_{(h,\phi)}$ for $h(x) = x$ and $\phi(x) = \frac{x^\alpha-x}{1-\alpha}$ and hence is a \typeone $(h,\phi)$-entropy. 

\paragraph{Previously known continuity bounds for Tsallis entropies}

Raggio \cite[Lemma 2]{Rag95} showed that $T_\alpha$ is Lipschitz continuous for $\alpha>1$:
\begin{equation} \label{eq:Tsallis-Lip}
|T_\alpha(\rho) - T_\alpha(\sigma)| \leq \frac{2\alpha}{\alpha-1} \eps
\end{equation}
if $T(\rho,\sigma) \leq \eps$,
while Zhang \cite[Theorem 1]{Zha07a} proved that if $T(\rho,\sigma)\leq \eps$ and $\alpha>1$, then
\begin{equation}
|T_\alpha(\rho) - T_\alpha(\sigma)| \leq \begin{cases}
 \frac{1}{1-\alpha}(\eps^\alpha (d-1)^{1-\alpha} +  (1-\eps)^\alpha -1)  & \eps < 1 - \frac{1}{d}\\
 \frac{d^{1-\alpha} - 1}{(1-\alpha)} & \eps \geq 1 - \frac{1}{d}
\end{cases}
\label{eq:Tsallis-bound}
\end{equation}
using a coupling technique\footnote{In fact, \cite[Theorem 1]{Zha07a} considers the case $T(\rho,\sigma)=\eps$; in \eqref{eq:Tsallis-bound}, their bound has been monotonized to hold for $T(\rho,\sigma)\leq \eps$.}. In fact, \eqref{eq:Tsallis-bound} also holds in the case  $0<\alpha<1$ as was shown by \cite[(A.2)]{Audenaert07} via a direct optimization method (adapting the proof of \eqref{eq:Audenaert-Fannes_bound}). This bound for all $\alpha\in (0,1)\cup(1,\infty)$ also appears as \cite[Lemma 1.2]{Renyi-CMNF}, whose proof appears to follow the same direct optimization method as Audenaert. Zhang  \cite[Remark 4]{Zha07a} also derived \eqref{eq:Tsallis-Lip} from \eqref{eq:Tsallis-bound}, and \eqref{eq:Audenaert-Fannes_bound} from the limit $\alpha\to 1$ of \eqref{eq:Tsallis-bound}.

 \cite[Theorem 2.4]{FYK07} showed that if $\alpha\in [0,2]$ and $p,r\in \cP$ such that $\TV(p,r)= \eps \leq \alpha^{1/(1-\alpha)}$, then
\[
|T_\alpha(p) - T_\alpha(r)| \leq (2\eps)^\alpha \ln_\alpha(d) + \eta_\alpha(2\eps)
\]
where $\eta_\alpha(x) = - x^\alpha \ln_\alpha(x)$ and $\ln_\alpha(x) = \frac{x^{1-\alpha}-1}{1-\alpha}$. This bound is less tight than \eqref{eq:Tsallis-bound}, however.

\paragraph{Previously known continuity bounds for R\'enyi entropies}

Audenaert proved a tight uniform continuity bound for the $\alpha$-R\'enyi entropies for $\alpha \in (0,1)$ in 2007 (\cite{Audenaert07}), namely
for  $\eps\in [0,1]$ and $\rho,\sigma\in \mathcal{D}$ with $T(\rho,\sigma)\leq \eps$,
	\begin{equation}
 | H_\alpha(\rho) - H_\alpha(\sigma) | \leq \begin{cases}
\frac{1}{1-\alpha} \log ( (1-\eps)^\alpha + (d-1)^{1-\alpha} \eps^\alpha) & \eps < 1-\frac{1}{d}\\
\log d& \eps \geq 1 - \frac{1}{d},
\end{cases} \label{eq:uniform_Renyi_bound}
 \end{equation}

See \eqref{eq:bound-Rastegin} and \eqref{eq:Chen-renyi-bound} of \Cref{sec:cty-from-majflow} for the previously known continuity bounds in the case $\alpha > 1$. Note also \cite[Theorem 7]{WH19} provides continuity bounds for the $\alpha$-R\'enyi-entropy, although in the case $\alpha \in (0,1)$ bounds are not optimal (in contrast to \eqref{eq:uniform_Renyi_bound}), and in the case $\alpha>1$, are not as tight as the bounds presented here.

\paragraph{New continuity bounds}

\Cref{thm:hphi-GCB} provides an alternate proof of \eqref{eq:Tsallis-bound} for any $\alpha>0$ and of \eqref{eq:uniform_Renyi_bound} for $\alpha\in(0,1)$ and establishes that in either case for equality to occur, it is necessary and sufficient for one state to be pure, and the other state to have spectrum $\{1-\eps, \frac{\eps}{d-1}, \dotsc, \frac{\eps}{d-1})$ if $\eps < 1-\frac{1}{d}$, or $\{\frac{1}{d},\dotsc,\frac{1}{d}\}$ if $\eps \geq 1 - \frac{1}{d}$.

\begin{proposition} \label{cor:Tsallis-Lipschitz}
The $\alpha$-Tsallis entropies are Lipschitz continuous for $\alpha > 1$, with optimal Lipschitz constant $\frac{\alpha}{\alpha-1}$.
\end{proposition}
\begin{proof}	
\[
\Gamma_{T_\alpha}(r) = \frac{\alpha}{1-\alpha} (r_-^{\alpha-1} -r_+^{\alpha-1}) = \frac{\alpha}{\alpha-1} (r_+^{\alpha-1} -r_-^{\alpha-1}) \leq \frac{\alpha}{\alpha-1}
\]
with equality achieved by $r=  (1,0,\dotsc,0)$.
\end{proof}
\begin{remark}
This improves upon \eqref{eq:Tsallis-Lip} by a factor of 2, but can also be derived directly from \eqref{eq:Tsallis-bound}.
\end{remark}

The following is a corollary of \Cref{thm:Convex-type-Lip}.
\begin{proposition} \label{cor:Renyi-lipschitz}
The $\alpha$-R\'enyi entropy is Lipschitz continuous if and only if  $\alpha> 1$. In the latter case, the
optimal Lipschitz constant $k_\alpha$ satisfies
\begin{equation} \label{eq:Lip-estimate-Renyi}
\frac{\alpha}{\alpha-1}(d-2)^{1 - 1/\alpha}\frac{1}{2\ln(2)} \leq k_\alpha \leq \frac{d\alpha}{\alpha-1}\frac{1}{\ln(2)}.
\end{equation}
For certain values of $\alpha$, we compute $k_\alpha$ exactly or provide tighter bounds. We have
\begin{align}
k_\infty &= \frac{1}{\ln(2)} d,\\
k_2 &= \begin{cases}
 \frac{2}{\ln(2)} & d = 2\\ 
 \frac{d-2}{\sqrt{d-1}-1}\frac{1}{\ln(2)} & d > 2,
 \end{cases}
\end{align}
and for $\alpha\in(1,2)$,
\[
k_\alpha \leq \frac{\alpha}{\alpha-1} \frac{d^{\alpha-1}}{\ln(2)} < \frac{d\alpha}{\alpha-1}\frac{1}{\ln(2)}.
\]

Additionally, for $\alpha \in (0,1)$ and any $\delta\in (0,1]$, the smoothed entropy $H_\alpha^\delta$ is Lipschitz continuous, with optimal Lipschitz constant
\[
\frac{\alpha}{1-\alpha} \frac{1}{\ln(2)} \frac{ (\tfrac{\delta}{d-1})^{\alpha-1} - (1-\delta)^{\alpha-1} }{(1-\delta)^\alpha + (d-1)^{\alpha-1}\delta^\alpha}.
\]
\end{proposition}
\begin{remark}
\cite[Theorem 7 (2)]{WH19} can be used to establish the bound $k_\alpha \leq 2\frac{d\alpha}{\alpha-1}\frac{1}{\ln(2)}$ by a different technique.
\end{remark}
\begin{proof}	
The fact that the $\alpha$-R\'enyi entropies are not Lipshitz for $\alpha \leq 1$ follows from the fact that $g(\eps)$ defined in \eqref{eq:def_g} has $\frac{g(\eps)}{\eps}\xrightarrow{\eps\to0}\infty$.

Let us prove \eqref{eq:Lip-estimate-Renyi}; consider $\alpha > 1$. We have $h(x) = \frac{1}{1-\alpha}\log x$ and $\phi(x) =x^\alpha$. Then for $\alpha > 1$, we have $\phi(x)\geq 0$ and $\phi'(x)\geq 0$ for $x\in [0,1]$. Hence, since $-h'$ is strictly decreasing,
\begin{align*}	
 -h'(\phi(y) + (d-2)\phi(z) + \phi(x))(\phi'(y) - \phi'(x)) &\leq - h'(\phi(y))(\phi'(y) - \phi'(x)) \\
 &\leq - h'(\phi(y))\phi'(y)  \\
 &= - (h\circ \phi)'(y).
\end{align*}
Then
\[
 - (h\circ \phi)'(y) = \frac{\alpha}{\alpha-1}\frac{1}{\ln(2)}\frac{1}{y} \leq \frac{\alpha}{\alpha-1}\frac{1}{\ln(2)} d
\]
since $y \geq \frac{1}{d}$. 
Next, consider the lower bound.
Let $x=0$, $y = \frac{1}{(d-2)^{1 - 1/\alpha}}=\frac{1}{(d-2)^{(\alpha-1)/\alpha}}$, $z = \frac{1-y}{d-2}$, and let $r = (x,z,\dotsc,z,y) \in \cP^\downarrow$. Then
\begin{align*}	
\Gamma_{H_\alpha}(r) &= \frac{\alpha}{\alpha-1} \frac{1}{\ln(2)}\frac{y^{\alpha-1} - x^{\alpha-1}}{x^\alpha + y^\alpha + (d-2) z^\alpha} \\
&= \frac{\alpha}{\alpha-1} \frac{1}{\ln(2)}\frac{y^{\alpha-1}}{y^\alpha + (d-2) z^\alpha} \\
&= \frac{\alpha}{\alpha-1} \frac{1}{\ln(2)} \frac{y^{\alpha-1}}{y^\alpha + (1-y)^\alpha (d-2)^{1-\alpha}} \\
&= \frac{\alpha}{\alpha-1} \frac{1}{\ln(2)}\frac{1}{y} \frac{1}{1 + (1/y-1)^\alpha (d-2)^{1-\alpha}} \\
&\geq \frac{\alpha}{\alpha-1} \frac{1}{\ln(2)}\frac{1}{y} \frac{1}{1 + y^{-\alpha} (d-2)^{1-\alpha}} \\
&\geq \frac{\alpha}{\alpha-1} \frac{1}{\ln(2)}\frac{1}{y} \frac{1}{1 + 1} \\
&= \frac{\alpha}{\alpha-1} \frac{1}{2\ln(2)}(d-2)^{1 - 1/\alpha}.
\end{align*}

The proof for $\alpha=2$ is in \Cref{sec:H2}, and for $\alpha=\infty$ is in the next proposition. The proof for $\alpha \in(1,2)$ follows from \Cref{cor:Tsallis-Lipschitz} and \eqref{eq:bound-Renyi-by-Tsallis}. Lastly, for $\alpha < 1$ and $\delta > 0$, we find
\[
\lim_{\eps \to 0}\frac{g(\eps + \delta)}{\eps} = \frac{\alpha}{1-\alpha} \frac{1}{\log_\mathrm{e}(2)} \frac{ (\tfrac{\delta}{d-1})^{\alpha-1} - (1-\delta)^{\alpha-1} }{(1-\delta)^\alpha + (d-1)^(\alpha-1)\delta^\alpha}.
\]
from which the optimal Lipschitz constant follows by \Cref{prop:typ1-Lip}.
\end{proof}

For $\alpha=\infty$, we can obtain both a tight uniform continuity bound and the optimal Lipschitz continuity constant.
\begin{proposition} \label{prop:H-infinity}
$|H_\infty(p) - H_\infty(q)| \leq \log(1 + \eps d)$ if $\TV(p,q)\leq \eps$. In particular, $H_\infty$ has an optimal Lipschitz constant of $\frac{d}{\ln(2)}$.
\end{proposition}
\begin{proof}	
Since $H_\infty(r) = -\log r_+$, for $r\neq u$,  \eqref{eq:Gamma-for-symmetric-H} yields
\[
\Gamma_{H_\infty}(r) = (\partial_--\partial_+) H_\infty(r) = \partial_{+} \log r_+ = \frac{1}{\ln(2)}\frac{1}{r_+},
\]
whereas $\Gamma_{H_\infty}(u)=0$ as $\gen(u)=0$.
The optimal Lipschitz constant follows from the fact that $r_+ \geq \frac{1}{d}$ and for some probability vectors, $r_+ = \frac{1}{d}$. 

Moreover, $r \mapsto \frac{1}{\ln(2)}\frac{1}{r_+}$ is Schur concave (as the composition of a decreasing function and the Schur convex function $r\mapsto r_+$). Hence, $\Gamma_{H_\infty}$ is Schur concave on $\cP\setminus \{u\}$.
Since $\mmm_s$ is majorization-preserving for all $s\in[0,1]$, we have
\[
r\mapsto \Delta_\eps^{H_\infty} := H_\infty(\mmm_\eps(r)) -r = \int_0^\eps \Gamma_{H_\infty}(\mmm_s(r))\d s
\]
is Schur concave on $\cP\setminus B_\eps(u)$. This uses the fact that for $r\in \cP\setminus B_\eps(u)$, $\TV(r,u) > \eps$ and hence $\mmm_s(r)\neq u$ for all $s\in[0,\eps]$.
For any $r\in \cP\setminus B_\eps(u)$, $\mmm_t(r) \prec r$ for $t = \TV(r,u)$. Hence,
\[
\max_{r\in \cP}H_\infty(\mmm_\eps(r)) -r = \max_{r\in B_\eps(u)}H_\infty(\mmm_\eps(r)) -r = \max_{r\in B_\eps(u)} \log (d) - H_\infty(r)
\]
using that $\mmm_\eps(r) = u$ for $r\in B_\eps(u)$. Then
\[
\max_{r\in B_\eps(u)} \log (d) - H_\infty(r) = \max_{r\in B_\eps(u)} \log (d r_+).
\]
For $r\in B_\eps(u)$, $r_+ \leq \frac{1}{d}+\eps$, with equality for $r =  (\frac{1}{d}+\eps, \frac{1}{d},\dotsc, \frac{1}{d},\frac{1}{d}-\eps)$. Hence, putting it all together,
\[
\max_{r\in \cP}H_\infty(\mmm_\eps(r)) -r =  \log \left(d \left(\frac{1}{d}+\eps\right)\right) = \log (  1 + d\eps). \qedhere
\]
\end{proof}
\begin{remark}
The proof of \Cref{prop:H-infinity} shows the Schur concavity of $\Delta_\eps^{H_\infty}$ on $\cP\setminus B_\eps(u)$. This contrasts strongly with the Schur \emph{convexity} of $\Delta_\eps^H$ on $\cP$ for all \typeone{} $(h,\phi)$ entropies $H$ proven in \Cref{thm:Delta-eps-Schur-convex}.
\end{remark}

\subsection{R\'enyi entropy of parameter \texorpdfstring{$\alpha =1$}{alpha = 1} has optimal dimensional scaling in its continuity bound}\label{sec:Renyi_opt_scaling}

Note that the R\'enyi entropy of parameter $\alpha=1$ is the Shannon entropy (or von Neumann entropy, in the quantum case). There is a sense in which the continuity properties as a function of dimension $d$ of the R\'enyi entropy $H_\alpha$ are much improved at $\alpha=1$ compared to $\alpha \neq 1$. Let us introduce some notation. Define
\[
C_\alpha(d,\eps) := \sup_{\substack{p,q\in \cP_d\\ \TV(p,q)\leq \eps}} |H_\alpha(p) - H_\alpha(q)|
\]
as the optimal uniform continuity bound for $H_\alpha$ over $\cP_d$. Consider a sequence $(\eps_d)_{d\in \mathbb{N}}$ such that $\eps_d \xrightarrow{d\to\infty} 0$. Clearly, if
\begin{equation}\label{eq:Calpha_to_zero}
\limsup_{d\to\infty} C_\alpha(d, \eps_d) =0
\end{equation}
then for any sequences of distributions $(p_d)_{d\in \mathbb{N}}$ and $(q_d)_{d\in \mathbb{N}}$ with $p_q,q_d\in \cP_d$ such that $\TV(p_d,q_d) \leq \eps_d$, we have
\[
\limsup_{d\to\infty} |H_\alpha(p_d) - H_\alpha(q_d)| =0.
\]
Thus, any $(\eps_d)_{d\in \mathbb{N}}$ satisfying \eqref{eq:Calpha_to_zero} provides a dimensionally-aware notion of continuity for $H_\alpha$. Moreover, the slower $\eps_d$ converges to zero, the stronger the statement of continuity provided by \eqref{eq:Calpha_to_zero}. The following proposition therefore demonstrates that the case $\alpha=1$ is the ``most continuous'' in this sense.
\begin{proposition}
For any $s > 0$,
\begin{equation} \label{eq:C1_invpoly_zero}
C_1(d, d^{-s}) \xrightarrow{d\to\infty} 0.
\end{equation}
In fact, if $(\eps_d)_{d\in \mathbb{N}}$ has $\eps_d\log(d) \to 0$, then $\lim_{d\to\infty}C_1(d, \eps_d) = 0$.
However, for $\alpha \in(0,1)\cup(1,\infty)$,
\begin{equation} \label{eq:Calpha_invpoly_nonzero}
\liminf_{d\to\infty}C_\alpha(d, d^{-\frac{|\alpha-1|}{\alpha}}) > 0.
\end{equation}
and likewise $\liminf_{d\to\infty}C_\infty(d, d\inv) > 0$.
\end{proposition}
\begin{remark}
Note that this result contrasts with the scenario of fixed dimension $d$. The fact that $H_\alpha$ is Lipschitz continuous on $\cP_d$ if and only if $\alpha > 1$ provides a notion in which $H_\alpha$ is ``more continuous'' for $\alpha > 1$ than for $\alpha > 0$. In other words, if $\alpha > 0$, the only constant $k$ satisfying $|H_\alpha(p) - H_\alpha(q) | \leq k \TV(p,q)$ is $k=\infty$, while for $\alpha > 1$, finite $k$ suffices. This notion is not ``dimensionally-aware'', however, in the sense that $d$ is fixed.

Additionally, the parameter $\alpha=1$ for the Tsallis entropy, which again coincides with the Shannon entropy, does not admit optimal scaling of the continuity bound out of the whole family of Tsallis entropies; in fact, \eqref{eq:Tsallis-bound} shows that the optimal bound actually decreases with dimension for $\alpha>1$. This is related to the nonextensivity of the Tsallis entropy; extensivity enforces the scaling
\[
 H_\alpha(p^{\otimes n}) - H_\alpha(q^{\otimes n}) = n(H_\alpha(p) - H_\alpha(q))
\]
and since
\[
\TV(p^{\otimes n}, q^{\otimes n}) \leq n\TV(p,q)
\]
we find that
\[
C_\alpha(d^n, n \eps_d) \geq n C_\alpha (d,\eps_d),
\]
or, for $d = 2^n$,
\[
C_\alpha(d, \eps_d\log(d)) \geq \log(d) C_\alpha(2,\eps_d).
\]
Hence for any extensive entropy, the optimal continuity bound must grow at least logarithmically with dimension (up to the modification $\eps_d \leadsto \log(d) \eps_d$).
\end{remark}
\begin{proof}	
For $0\leq \eps < 1 - \frac{1}{d}$, the Audenaert-Fannes bound \eqref{eq:Audenaert-Fannes_bound} gives
\[
C_1(d, \eps) = h_2(\eps) + \eps\log(d-1)
\]
where $h_2(\eps) = -\eps\log(\eps) - (1-\eps)\log(1-\eps)$ is the binary entropy. Hence, if $\eps_d \log(d) \to 0$, then $\eps_d \to 0$ and
\[
C_1(d, \eps_d) = h_2(\eps_d) + \eps_d \log(d-1) \to 0
\]
as well. Since for any $s>0$ we have $d^{-s} \log(d) \to 0$, we recover \eqref{eq:C1_invpoly_zero}.

Next, let us establish \eqref{eq:Calpha_invpoly_nonzero} for $0<\alpha<1$. In this case, for $0\leq \eps < 1 - \frac{1}{d}$ \eqref{eq:uniform_Renyi_bound} gives
\[
C_\alpha(d, \eps) = \frac{1}{1-\alpha} \log ( (1-\eps)^\alpha + (d-1)^{1-\alpha} \eps^\alpha).
\]
For $\eps_d = d^{-\frac{1-\alpha}{\alpha}}$, we have $(d-1)^{1-\alpha} \eps_d^\alpha \to 1$ while $(1-\eps_d)^\alpha\to 0$, and hence $C_\alpha(d,\eps_d) \to \frac{1}{1-\alpha} > 0$.

For $\alpha > 1$, we do not have an exact expression for $C_\alpha(d,\eps)$. However, since a Lipschitz bound provides a uniform continuity bound, we have
\[
C_\alpha(d, \eps) \geq \eps \sup_{r\in \cP_+} \Gamma_{H_\alpha}(r)
\]
using \Cref{cor:opt-Lip}. Then \Cref{cor:Renyi-lipschitz} shows that $k_\alpha = \sup_{r\in \cP_+} \Gamma_{H_\alpha}(r)$ satisfies
\[
k_\alpha \geq \frac{\alpha}{\alpha-1}(d-2)^{1 - 1/\alpha}\frac{1}{2\ln(2)}.
\]
Hence, if $\eps_d = d^{-\frac{\alpha-1}{\alpha}}$, then
\[
\eps_d(d-2)^{1 - 1/\alpha} = \eps_d (d-2)^{(1-\alpha)/\alpha} \to 1.
\]
Therefore $\liminf_{d\to\infty}  C_\alpha(d, \eps) \geq  \frac{\alpha}{\alpha-1}\frac{1}{2\ln(2)}>0$.
Likewise, for $\alpha=\infty$, the optimal Lipschitz constant is given in \Cref{prop:H-infinity} as $k_\infty = \frac{d}{\ln(2)}$, and hence $\liminf_{d\to\infty}C_\infty(d,d\inv) \geq \frac{1}{\ln(2)}>0$.
\end{proof}

\subsubsection{Discussion of previous continuity bounds for \texorpdfstring{$H_\alpha$}{the alpha-Renyi entropy} with \texorpdfstring{$\alpha > 1$}{alpha > 1}} \label{sec:discuss-prev-bounds}

As mentioned at the start of the section, continuity bounds on the $\alpha$-R\'enyi entropy for $\alpha > 1$ were not known until 2011, and the bounds known until now have poor scaling $\sim d^{\alpha-1}$. In this section, we use majorization flow as a tool to understand why the previous bounds performed poorly. The technique used to establish the previous bounds was to relate the difference in R\'enyi entropy of two distributions to the corresponding difference in Tsallis entropy \cite{Ras2011,Renyi-CMNF}. In the following proposition, we show that even a relaxed version of this pointwise comparison must necessarily yield a bound scaling as $d^{\alpha-1}$. In contrast, we show that the comparison between the \emph{maximum} difference of the two entropies exhibits much better scaling, and in fact can yield the Lipschitz continuity bound of \eqref{eq:Lip-estimate-Renyi}. This can be understood by the fact that the Tsallis entropy is a \typeone{} $(h,\phi)$-entropy, and hence it increases the slowest near the uniform distribution (in the sense that \Cref{thm:Delta-eps-Schur-convex} holds). On the other hand, for $\alpha>1$, the $\alpha$-R\'enyi entropy is a \typetwo{} $(h,\phi)$-entropy, and increases  quickest at a distribution of the form $(x,z,\dotsc,z,y)$ for $x\leq z \leq y$ as shown by \Cref{thm:Convex-type-Lip}, which can be close to uniform.

\begin{proposition} \label{prop:compare-Renyi-Tsallis}
The smallest constant $c$ such that
\begin{equation} \label{eq:bound-Renyi-by-Tsallis-prointwise}
\max_{q \in B_\eps(r)}(H_\alpha(q) - H_\alpha(r)) \leq c \max_{q \in B_\eps(r)}(T_\alpha(q) - T_\alpha(r)) 
\end{equation}
for all $r\in \cP$ and $\eps \in (0,1]$ is $c = \frac{d^{\alpha-1}}{\ln(2)}$. However, the smallest constant $\tilde c$ such that
\begin{equation} \label{eq:bound-Renyi-by-Tsallis}
\max_{r \in \cP}\max_{q \in B_\eps(r)}(H_\alpha(q) - H_\alpha(r)) \leq \tilde c \max_{r' \in \cP} \max_{q \in B_\eps(r')}(T_\alpha(q) - T_\alpha(r')) 
\end{equation}
for all $\eps \in (0,1]$ satisfies $\tilde c \leq \frac{\alpha d}{\ln(2)}$.
\end{proposition}
\begin{proof}	
Since
\[
\max_{q \in B_\eps(r)}(H_\alpha(q) - H_\alpha(r)) = \int_0^\eps \Gamma_{H_\alpha}(\mmm_\eps(r)) \d s
\]
and
\[
\max_{q \in B_\eps(r)}(T_\alpha(q) - T_\alpha(r))  = \int_0^\eps \Gamma_{T_\alpha}(\mmm_\eps(r)) \d s
\]
it suffices to bound the ratio $\frac{\Gamma_{H_\alpha}(p)}{\Gamma_{T_\alpha}(p)}$ uniformly in $p\in \cP$ by $c$. On the other hand, for \eqref{eq:bound-Renyi-by-Tsallis} to hold for all $\eps \in (0,1]$ and $r\in \cP$, the same ratio must in fact be bounded by $c$.  We have
\[
\Gamma_{H_\alpha}(p) =\frac{1}{\ln(2)} \frac{\alpha}{\alpha-1} \frac{r_+^{\alpha-1} - r_-^{\alpha-1}}{\sum_{i=1}^d p_i^\alpha}, \qquad \Gamma_{T_\alpha}(p) = \frac{\alpha}{\alpha-1} (r_+^{\alpha-1} - r_-^{\alpha-1})
\]
and hence
\[
\frac{\Gamma_{H_\alpha}(p)}{\Gamma_{T_\alpha}(p)} =\frac{1}{\ln(2)} \left( \sum_{i=1}^d p_i^\alpha \right)^{-1}.
\]
Since $p\mapsto \sum_{i=1}^d p_i^\alpha$ is Schur convex, the above ratio is Schur concave, and hence maximized at the uniform distribution. Thus,
\[
\max_{p\in \cP}\frac{\Gamma_{H_\alpha}(p)}{\Gamma_{T_\alpha}(p)} =\frac{1}{\ln(2)} \left( d \left(\frac{1}{d}\right)^\alpha\right)^{-1} =\frac{1}{\ln(2)} d^{\alpha-1}.
\]
To estimate $\tilde c$, we simply rewrite the uniform continuity bound for $T_\alpha$ given in \Cref{eq:hphi-uniform-bound} as:
\[
 \max_{r' \in \cP} \max_{q \in B_\eps(r')}(T_\alpha(q) - T_\alpha(r')) = \frac{1}{\alpha-1} ( 1 - (1-\eps)^\alpha - (d-1)^{1-\alpha}\eps^\alpha ),
\]
noting that the maximum is achieved at $r' = (1,0,\dotsc,0)$. We have $(1-\eps)^\alpha  \leq 1-\eps$ and hence $1-(1-\eps)^\alpha \geq \eps$. Then
\[
 \max_{r' \in \cP} \max_{q \in B_\eps(r')}(T_\alpha(q) - T_\alpha(r')) \geq  \frac{ \eps}{\alpha-1}.
\]
On the other hand, 
\[
\max_{r \in \cP}\max_{q \in B_\eps(r)}(H_\alpha(q) - H_\alpha(r))  \leq \eps\frac{\alpha}{\alpha-1}\frac{d}{\ln(2)}
\]
by \eqref{eq:Lip-estimate-Renyi}.
\end{proof}

As mentioned in the Introduction, our technique for proving entropic continuity bounds by using the notion of majorization flow has the additional advantage of providing an insight into the entropy landscape with respect to the TV distance. The above proposition shows how majorization flow can be used to better understand previous methods for establishing continuity bounds, in addition to establishing new ones.

\subsubsection{Connection to thermodynamics} \label{sec:connect-thermo}
\cite{Baez11} introduced an interesting connection between the $\alpha$-R\'enyi entropy and free energies in thermodynamics. In this section, we recall this relationship, and remark on the resulting consequences of our continuity bounds for R\'enyi entropies of Gibbs states and their connection to changes in free energy.

The following holds in either a quantum or classical picture. We will work in quantum notation for consistency with \cite{Baez11}. Consider a Gibbs state
\[
\rho(T) = Z(T)\inv\e^{-H/T}
\]
where $H$ is the Hamiltonian, $T \geq 0$ the temperature, $Z(T) = \tr(\e^{-H/T})$ is the partition function, and we have set Boltzmann's constant $k_\text{B}\equiv 1$. We can define the free energy as
\[
F(T) = - T \ln Z(T).
\]
By direct calculation, we find that the $\alpha$-R\'enyi entropy $H_\alpha$ satisfies
\begin{equation} \label{eq:Renyi-entropy-as-free-energies}
    H_{\frac{T_0}{T}}(\rho(T_0)) = -\frac{F(T)- F(T_0)}{T-T_0}.
\end{equation}
for any $T > 0$ \cite[Equation (9)]{Baez11}. In the limit $T\to T_0$, we recover the thermodynamic relation
\[
 H(\rho(T_0)) = - \left.\frac{\d F}{\d T}\right|_{T_0}
 \]
 that the entropy is the derivative of the free energy with respect to temperature. Note that any full-rank state $\sigma$ can be seen as a Gibbs state at temperature $T$ associated to the Hamiltonian $H = -\frac{1}{T}\log \sigma$. This gives a physical interpretation to $H_\alpha(\sigma)$ for any full-rank state $\sigma$: consider $\sigma$ to be the Gibbs state at initial temperature $T_0=1$. Then $H_\alpha(\sigma)$ equals the negative of the ratio of the change in free energy to the change in temperature when the temperature is changed from $T_0$ to $\alpha\inv T_0$. This can be seen as the maximum amount of work the system, initially in thermal equilibrium at temperature $T_0$, can do when its temperature is suddenly changed from $T_0$ to $\alpha\inv T_0$ as it moves to the new thermal equilibrium, divided by the change in temperature \cite{Baez11}.

The Schur concavity of $H_\alpha$ for all $\alpha > 0$ can be interpreted through this physical picture as well. The relation $\rho \prec \sigma$ means that the distribution of the eigenvalues of $\rho$ is ``flatter'' and ``more disordered'' than those of $\sigma$; correspondingly, $H_\rho := - \frac{1}{T}\log \rho$ log-majorizes $H_\sigma := - \frac{1}{T}\log \sigma$, where log majorization is defined by $A \prec_{\log} B$ if $\log A \prec \log B$. By Schur concavity, if $\rho \prec \sigma$, then $H_\alpha(\rho)\geq H_\alpha(\sigma)$ for any $\alpha > 0$. Hence, the Schur concavity of the $\alpha$-R\'enyi entropy can be interpreted as a statement about how the distribution of energy levels of a Hamiltonian relates to the free-energy increase or decrease of the system (per unit change in temperature) under a sudden change in temperature.

Now, recall that the trace distance is endowed with an operational interpretation in terms of distinguishability under measurement. We say that $\rho$ and $\sigma$ are $\eps$-indistinguishable if $\frac{1}{2}\|\rho-\sigma\|_1 \leq \eps$. Consider an experiment in which the system is in a state $\sigma$ which is not known precisely, but is $\eps$-indistinguishable from a known state $\rho$, which is a Gibbs state, $\rho = \rho(T_0)$. We pose the following question:

\smallskip
If the temperature is abruptly changed from $T_0\to T$, can one bound the ratio of the change in free energy and the corresponding change in temperature?
\smallskip

This quantity is exactly the $\alpha$-R\'enyi entropy of order $\alpha = T_0/T$ by \eqref{eq:Renyi-entropy-as-free-energies}, and hence \eqref{eq:uniform_Renyi_bound} and \Cref{cor:Renyi-lipschitz} provide an answer in the affirmative. Moreover, \Cref{cor:Renyi-lipschitz} shows that the $\alpha$-R\'enyi entropy is Lipschitz continuous if and only if $\alpha > 1$. In other words, if $T_0 > T$, then there exists $k_{T_0/T} < \infty$ such that
\[
|H_{T_0/T}(\sigma) - H_{T_0/T}(\rho(T_0))| \leq \eps k_{T_0/T}.
\]
If $T_0 \leq T$, then no such linear bound can hold uniformly in $\sigma$ and $\rho(T_0)$, but \eqref{eq:hphi-uniform-bound} gives a tight uniform (nonlinear) bound. 

\subsection{Continuity bounds for other entropies} \label{sec:more-cty-bounds}

\begin{itemize}
	
	\item The von Neumann entropy
	\[
	 H(\rho) = -\tr (\rho \log \rho).
	 \] 
	 $H = H_{(h,\phi)}$ for $h(x) = x$ and $\phi(x) = -x \log x$, which is a \typeone{} $(h,\phi)$-entropy. The von Neumann entropy satisfies the following tight continuity bound: Given $\eps\in(0,1]$ and $\rho,\sigma\in \cD(\cH)$ with $T(\rho,\sigma)\leq \eps$, 
	 \begin{equation}
 |H(\rho) - H(\sigma) | \leq \begin{cases}
 \epsilon \log (d-1) + h(\epsilon) & \text{if } \epsilon < 1 - \tfrac{1}{d} \\
	 \log d & \text{if } \epsilon \geq 1 - \tfrac{1}{d}
 \end{cases} \label{eq:Audenaert-Fannes_bound}
 \end{equation}
where $h(\eps) := - \eps \log \eps - (1-\eps) \log (1-\eps)$ denotes the binary entropy. This inequality is known in the quantum information theory literature as the Audenaert-Fannes bound, which is a strengthened version of the Fannes bound \cite{Fannes1973} and was established in \cite{Audenaert07} by a direct optimization argument. The bound was also proven via a coupling argument in \cite[Theorem 3.8]{PetzQITbook}, with credit to Csiszar. In the classical case, it was also proven by \cite[Eq. (4)]{Zha07} by a coupling argument, and by Ho and Yeung \cite[Theorem 6]{HY10} via an optimization over local continuity bounds, similar in spirit to the techniques used in this paper (although with an analysis specific to the Shannon entropy, as opposed to \typeone{} $(h,\phi)$-entropies more generally). The coupling argument was revisited in \cite[Theorem 3]{Sas13a} and \cite[Lemma 1]{Winter16}. In each case (except for \cite[Theorem 3.8]{PetzQITbook}), it was shown that equality occurs if one state is pure, and the other state has the spectrum $\{1-\eps,\frac{\eps}{d-1},\dotsc,\frac{\eps}{d-1}\}$. These conditions were shown to be necessary in \cite{HD17} by an analysis of the coupling argument. 

\textbf{Our contribution} \Cref{thm:hphi-GCB} provides an alternate proof for \eqref{eq:Audenaert-Fannes_bound} and the necessary and sufficient conditions for equality as a consequence of the fact that the von Neumann entropy is a \typeone{} $(h,\phi)$-entropy, and \Cref{prop:typ1-Lip} shows that $H$ is not Lipschitz continuous with respect to the trace distance, but its smoothed variant $H^\delta$ is Lipschitz continuous for any $\delta \in (0,1)$, with an optimal Lipschitz constant of
\[
\log ( \delta\inv - 1) + \log(d-1).
\]

\item The $(s,\alpha)$-unified entropies,
	\[
	E_\alpha^s(\rho) = \frac{1}{s(1-\alpha)} (\tr[\rho^\alpha]^s - 1)
	\]
	for $\alpha\in (0,1)\cup (1,\infty)$, $s\in \mathbb{R}\setminus\{0\}$, were introduced in the quantum case by \cite{HY06} and in the classical case by \cite{RT91}. This family of entropies includes the Tsallis entropies in the case $s=1$, and the $\alpha$-R\'enyi entropies (up to a factor of $\ln(2)$) in the limit $s\to 0$. We have $E_\alpha^s = H_{(h,\phi)}$ for $\phi(x) = x^\alpha$, and $h(x) = \frac{1}{s(1-\alpha)}(x^s - 1)$, which satisfy $\phi(0)=0$, and $h(\phi(1))=0$. If $\alpha\in (0,1)$, $h$ is strictly increasing and $\phi$ is strictly concave, while if $\alpha >1$, $h$ is strictly decreasing, and $\phi$ is strictly convex. Additionally, $h$ is convex if ($s>1$ and $\alpha <1$) or if ($s<1$ and $\alpha >1$), and is concave otherwise. Thus, if $0<\alpha < 1$ and $s\leq 1$, $E_\alpha^s$ is a \typeone{} $(h,\phi)$-entropy, and if $\alpha > 1$ with $s\leq 1$, then $E_\alpha^s$ is a \typetwo{} $(h,\phi)$-entropy. 
	If $s>1$, then for any $\alpha\in (0,1)\cup(1,\infty)$,  $E_\alpha^s$ is an $(h,\phi)$-entropy in the sense defined by \cite{SMMP93a}, but not of \typeone{} or \typetwo{}, and hence the results of \Cref{sec:cty-from-majflow} do not apply in that case.
	\begin{remark}
	 \cite[Proposition 5]{HY06} incorrectly claims that the unified entropies are not Schur concave. However, they are indeed strictly Schur concave for all $\alpha\in (0,1)\cup (1,\infty)$ and $s\in \mathbb{R}\setminus\{0\}$.
	\end{remark}
	Rastegin \cite{Ras2011} showed that for $0<\alpha<1$, $s\in (-\infty,-1]\cup[0,1]$, and $T(\rho,\sigma) \leq \frac{1}{2}\alpha^{\frac{1}{1-\alpha}}$, the bound
	\begin{equation}
	|E_\alpha^s(\rho) - E_\alpha^s(\sigma)| \leq (2 \eps)^\alpha \ln_\alpha d + n_\alpha(2\eps)
	\end{equation}
	holds, where $\ln_\alpha = \frac{x^{1-\alpha} - 1}{1-\alpha}$ and $n_\alpha(x) = \frac{x^\alpha - x }{1-\alpha}$.
    If $\alpha > 1$ and $s\in[-1,0]\cup [1,+\infty]$,
	\begin{equation}
	|E_\alpha^s(\rho) - E_\alpha^s(\sigma)| \leq \chi_s [ \eps^\alpha \ln_\alpha(d-1) + t_\alpha(\eps) ]
	\end{equation}
	where $t_\alpha(\eps) := T_\alpha(\{\eps, 1-\eps\})$ is the binary Tsallis entropy. In \cite[Proposition 6]{HY06}, the Lipschitz continuity bound
\begin{equation}
|E_\alpha^s(\rho) - E_\alpha^s(\sigma)| \leq \frac{\alpha}{\alpha-1} \eps
\end{equation}
for any $\rho,\sigma\in \mathcal{D}(\cH)$ with $T(\rho,\sigma)\leq \eps$ for $\alpha > 1$, and $s\geq 1$ was proven.

\textbf{Our contribution:}

For $\alpha\in (0,1)$ and $s\leq 1$, then $E_\alpha^s$ is not Lipschitz continuous on $\cD(\cH)$ by \Cref{prop:typ1-Lip}, but satisfies the following tight uniform continuity bound by \Cref{thm:hphi-GCB}: If $\eps\in [0,1]$ and $\rho,\sigma\in \mathcal{D}$ with $T(\rho,\sigma)\leq \eps$,
\begin{equation}
|E_\alpha^s(\rho) - E_\alpha^s(\sigma)| \leq \begin{cases}\frac{1}{s(1-\alpha)}
\left[\big( (1-\eps)^\alpha + (d-1)^{1-\alpha}\eps^\alpha \big)^s - 1\right] & \eps < 1-\frac{1}{d}\\
\frac{1}{s(1-\alpha)}\left[d^{s(1-\alpha)}-1\right] & \eps \geq 1 - \frac{1}{d}
\end{cases}
\end{equation}
 with equality if and only if one state is pure, and the other state has spectrum $\{1-\eps, \frac{\eps}{d-1}, \dotsc, \frac{\eps}{d-1}\}$ if $\eps < 1-\frac{1}{d}$, or $\{\frac{1}{d},\dotsc,\frac{1}{d}\}$ if $\eps \geq 1 - \frac{1}{d}$.

If $\alpha>1$ and $s\leq 1$, then $E_\alpha^s$ is Lipschitz continuous on $\cD(\cH)$ by \Cref{thm:Convex-type-Lip}, and the associated optimal Lipschitz constant $k_{\alpha}^s$ satisfies
\[
k_\alpha^s \leq \begin{cases}
\frac{\alpha}{\alpha-1}d^{1-\alpha s} &s \alpha < 1\\
\frac{\alpha}{\alpha-1} & s\alpha \geq 1.
\end{cases}
\]

\item Entropies induced by divergences. Denoting left multiplication by an operator $A$ as $L_A$, and right multiplication by $A$ as $R_A$, one defines the $f$-divergence
\begin{equation}
S_f(\rho\|\sigma) := \tr[ \sigma^{1/2} f(k_\rho R_{\sigma\inv}) (\sigma^{1/2})]
\end{equation}
which was first introduced by Petz \cite{Pet85} (see \cite{HM17a} for a useful overview). The \emph{maximal $f$-divergence} \cite{PR98} is given by
\begin{equation}
\hat S_f(\rho\|\sigma) := \tr[\sigma^{1/2} f(\sigma^{-1/2}\rho \sigma^{-1/2}) \sigma^{1/2}].
\end{equation}

From either divergence, one can define an associated entropy by evaluating at $\sigma=\one$ (and reversing the sign). The two entropies coincide, yielding
\begin{equation}
S_f(\rho) := - S_f(\rho\|\one) = - \hat S_f(\rho\||\one) = - \tr [ f(\rho)].
\end{equation}
For strictly convex $f$ with $f(0)=f(1)=0$ we can define $\phi = -f$ and $h(x) = x$, yielding a \typeone{} $(h,\phi)$-entropy.

\textbf{Our contribution:}
For strictly convex $f$ with $f(0)=f(1)=0$, \Cref{thm:hphi-GCB} gives that for $\eps\in [0,1]$ and $\rho,\sigma\in \cD(\cH)$ with $T(\rho,\sigma)\leq \eps$,
\begin{equation} \label{eq:unif-bound-f-div}
 | S_f(\rho) - S_f(\sigma) | \leq \begin{cases}
 -f(1-\eps) - (d-1) f( \frac{\eps}{d-1}) & \eps < 1-\frac{1}{d}\\
-d f(\frac{1}{d}) & \eps \geq 1 - \frac{1}{d}.
\end{cases} 
\end{equation}
\Cref{prop:typ1-Lip} shows that $S_f$ is Lipschitz continuous on $\cD(\cH)$ if and only if
\[
k := \lim_{\eps\to 0}-\frac{1}{\eps}f(1-\eps) - (d-1) \frac{1}{\eps}f( \frac{\eps}{d-1})
\]
is finite. In the latter case, $k$ is the optimal Lipschitz constant for $S_f$. Note that if $f$ is differentiable at $0$ and $1$, then $k = f'(1) - f'(0)$.
\item The \emph{concurrence} of a bipartite pure state $\psi_{AB}$ is an entanglement monotone defined as
\[
C(\psi_{AB}) = \sqrt{2(1 - \tr[\rho_{A}^2])}
\]
where $\rho^\psi = \tr_B[\psi_{AB}]$ is the reduced state on system $A$ \cite{Woo01,RBC+01}. Regarded as a function of the reduced state, the concurrence can be seen as \typeone{} $(h,\phi)$-entropy with $\phi(x) = -x^2$ and $h(x) = \sqrt{2(1+x)}$, and hence \Cref{thm:hphi-GCB} gives a tight uniform continuity bound in terms of the trace distance between the reduced states. If for some $\eps\in [0,1]$, two bipartite pure states $\psi_{AB}$ and $\phi_{AB}$ satisfy $T(\psi_{AB},\phi_{AB})\leq \eps$, then by monotonicity of the trace distance under partial trace, $T(\rho^\psi, \rho^\phi) \leq \eps$ as well. Hence, \Cref{thm:hphi-GCB} yields
\begin{equation}
 |C(\psi_{AB}) - C(\phi_{AB})| \leq \begin{cases}
 \sqrt{2( 1 - (1-\eps)^2 - (d-1)^{-1}\eps^2 )} & \eps < 1 - \frac{1}{d}\\
\sqrt{2(1-d\inv)} & \eps \geq 1 - \frac{1}{d}
 \end{cases}
\end{equation}
for all bipartite pure states $\psi_{AB}$ and $\phi_{AB}$ such that $T(\psi_{AB},\phi_{AB}) \leq \eps$. The concurrence is not Lipschitz continuous, by \Cref{prop:typ1-Lip}.
\end{itemize}

\section{A continuity bound for the number of distinct observations from $N$ independent trials} \label{sec:butterflies}
In this section, we consider another application of the majorization flow to continuity bounds.

Consider an experiment in which outcome $i \in \{1,\dotsc,M\}$ is observed with probability $p_i$, for some probability distribution $p\in \cP$. Repeat this experiment $N$ times, independently, and consider the random variable $K$ which denotes the number of distinct outcomes observed. In \cite{WY73} it was shown that each entry of the cumulative distribution, namely
\[
F_j \equiv F_j(p) := \Pr[ K \leq j],
\]
is a Schur convex function of $p$ (Theorem 4.1 of the above article), and has the expression
\[
F_j = \sum_{i=1}^j (-1)^{j-i}{M - i - 1 \choose j - i} \sum_{1 \leq l_1 < l_2 < \dotsm < l_i \leq M} (p_{l_1} + \dotsm + p_{l_i})^N
\]
which is given in \cite[Corollary 3.2]{WY73}. Taking the derivative, we obtain
\[
\frac{\partial F_j}{\partial p_1} - \frac{\partial F_j}{\partial p_2} = f_j(p_1) - f_j(p_2)
\]
for
\[
f_j(s)  =  N \sum_{i=1}^j (-1)^{j-i}{M - i - 1 \choose j - i} \sum_{\substack{1 < l_2 < \dotsm < l_i \leq M\\ l_2,\dotsc,l_i \neq 1,2}} (s + p_{l_2} + \dotsm + p_{l_i})^{N-1}
\]
as was calculated in \cite[Equation 11]{WY73}. To compute the optimal Lipschitz constant for $F_j$, it remains to maximize this difference over $p\in \cP$. We leave that for future work, and simply show that $F_j$ is a Lipschitz continuous function of $p\in \cP$ by showing that $f_j(p_1) - f_j(p_2) < \infty$ for any $p\in \cP$. We can use that the summand $(s + p_{l_2} + \dotsm + p_{l_i})$ is less than $1$ for $s\in \{p_1, p_2\}$, and that the number of elements in the second summation is ${M-2 \choose i}$ to find the simple bound
\[
|f_j(s)| \leq N \sum_{i=1}^j {M - i - 1 \choose j - i} {M-2 \choose i} < \infty
\]
for $s\in \{p_1, p_2\}$, which completes the proof.

Next, \cite[Corollary 3.3]{WY73} shows that the expected number of distinct elements, $\bE[K]$ satisfies
\[
\bE[K] \equiv \bE_p[K] = M- \sum_{i=1}^M (1 - p_i)^N.
\]
In fact, we can identify $\bE[K]-1$ as a \typeone{} $(h,\phi)$-entropy, with $h(x) = x-1$ and $\phi(x) = 1 - (1-x)^N$, as defined in \Cref{sec:cty-from-majflow}. Hence, \Cref{cor:tight-uniform-bounds-concave-type} shows that for $\eps >0$ if $p,q \in \cP$ satisfy $\TV(p,q) \leq \eps$, then
\[
|\bE_p[K] - \bE_q[K]| \leq \begin{cases}
\frac{(M-1)^N - (M-1-\varepsilon)^N}{(M-1)^{N-1}} - \varepsilon^N & \varepsilon \leq 1 - \frac{1}{M} \\ 
\frac{M^N - (M-1)^N}{M^{N-1}}-1 & \eps > 1 - \frac{1}{M}.
\end{cases}
\]
In particular, using \Cref{prop:typ1-Lip},
\begin{equation}\label{eq:trials-Lip}
|\bE_p[K] - \bE_q[K]| \leq \varepsilon N,
\end{equation}
and $N$ is the optimal Lipschitz constant. Note that \eqref{eq:trials-Lip} does not depend on $M$, but its derivation assumes $M < \infty$.

\paragraph{Acknowledgements}
The authors would like to thank Koenraad Audenaert for interesting discussions regarding bounding the Lipschitz constant for the $\alpha$-R\'enyi entropy in the case $\alpha > 1$. E.H. would like to thank George Hanson for introducing him to the notion of concurrence.
E.H. is supported by the Cantab Capital Institute for the Mathematics of Information (CCIMI).

\appendix

\section{Optimal Lipschitz constant for the collision entropy} \label{sec:H2}

It remains to calculate the right-hand side of \eqref{eq:Convex-type-Lip-d-gt-2} in the case that $h(x) = -\log x$ and $\phi(x) = x^2$.

\begin{itemize}

\item In the case $d=2$, $\alpha\in (1,2]$, we aim to maximize the function 
\[
h(x,y) = \frac{y^{\alpha-1} - x^{\alpha-1}}{x^\alpha + y^\alpha}
\] 
where $y=1-x$. For $\alpha\in (1,2]$, we have the bound
\begin{align*}	
h(x,y) \leq 1 &\iff y^{\alpha-1} - x^{\alpha-1} \leq y^\alpha + x^\alpha\\
&\iff y^{\alpha-1} - y^{\alpha} \leq x^{\alpha-1} + x^\alpha\\
&\iff y^{\alpha-1}(1-y) \leq x^{\alpha-1}(1-x)\\
&\iff y^{\alpha-1}x \leq x^{\alpha-1}y\\
&\iff y^{\alpha-2}\leq x^{\alpha-2}
\end{align*}
using $y=1-x$ and $x=1-y$. We find $h(x,y) \leq 1$ since $x\leq y$ and $\alpha\in (1,2]$. On the other hand, $h(0,1) = 1$ for all $\alpha$. 

\item For $d>2$, and $\alpha = 2$, we consider
\[
h(x,y) = \frac{y - x}{x^2 + y^2 + (d-2)z^2}
\]
where $z \equiv z(x,y) = \frac{1-x-y}{d-2}$.
Then for $D=x^2 + y^2 + (d-2)z^2$,
\begin{align*}	
D^2 \frac{\d}{\d x} h(x,y) &= -(x^2 + y^2 + (d-2)z^2) - (y-x)(2x - 2z)\\ 
&=-(x^2 + y^2 + (d-2)z^2) +2 (y-x)(z-x)
\end{align*}
so $\frac{\d}{\d x} h(x,y)\leq 0$ iff
\begin{gather*}
2 (y-x)(z-x) \leq x^2 + y^2 + (d-2)z^2\\
2(yz - xz + x^2 - xy)\leq x^2 + y^2 + (d-2)z^2\\
2(yz - xz - xy)\leq -x^2 + y^2 + (d-2)z^2\\
\end{gather*}
We have $2yz \leq y^2 + d^2 \leq y^2 + (d-2)z^2$, so it remains to show $-2(xz + xy) \leq - x^2$. But that follows from
\[
x^2 \leq 4x^2 \leq 2(xz + xy)
\]
using $x\leq y$ and $x\leq z$. Thus, for any $y \geq x$, such that $x\leq z \leq y$, $h(x,y)$ is decreasing in $x$.

Hence, we consider $x=0$. Then $z = \frac{1-y}{d-2}$, and $y$ is constrained only by $y\in[\frac{1}{d-1}, 1]$. Then
\[
h(0,y) = \frac{y}{y^2 + (d-2)\inv (1-y)^2} = \frac{y (d-2)}{ (d-2)y^2 + 1 - 2y + y^2} = \frac{(d-2)y}{(d-1)y^2 + 1 - 2y}.
\]
We have
\[
((d-1)y^2 + 1 - 2y)^2\partial_y h(0,y) = ((d-1)y^2 + 1 - 2y)(d-2) - (d-2)y ( 2(d-1)y - 2)
\]
so $\partial_y h(0,y) = 0$ if and only if
\begin{gather*}
(d-1)y^2 + 1 - 2y = y ( 2(d-1)y - 2)\\
(d-1)y^2 + 1 - 2y =  2(d-1)y^2 - 2y\\
 1  =  (d-1)y^2 \\
y = \frac{1}{\sqrt{d-1}}.
\end{gather*}
Note that with $x=0$, $y = \frac{1}{\sqrt{d-1}}$, we have $0\leq x\leq z\leq y\leq 1$ so the constraints are satisfied. It could be that this choice of $y$ yields only a local maximum. To rule this case out, since $y \in [ \frac{1}{d-1}, 1]$, we check $h(0, \frac{1}{d-1})  = h(0, 1) = 1$, while
\[
h( 0, (d-1)^{-1/2}) = \frac{d-2}{2\sqrt{d-1}-2} > 1
\]
for $d \geq 3$.
\end{itemize}

\section{An elementary property of concave functions} \label{sec:elem-prop-concave-fun}

Given a function $\phi: I\to \R$ defined on an interval $I\subset \R$, we define the ``slope function,''
\[
\cs(x_1,x_2) =  \frac{\phi(x_2)- \phi(x_1)}{x_2-x_1}
\]
for $x_1,x_2\in I$ with $x_1\neq x_2$.
Note that $\cs$ is symmetric in its arguments. It can be shown that $\phi$ is concave (resp.~strictly concave) if and only if $\cs$ is monotone decreasing (resp.~strictly decreasing) in each argument.
\begin{proposition} \label{cor:main_concave_ieq}
Let $I\subset \R$ be an interval and $\phi: I \to \R$ be concave. For any $x_1,x_2,y_1,y_2\in I$ such that $x_1\neq x_2$, $y_1\neq y_2$, $x_1 \leq y_1$ and $x_2 \leq y_2$
we have
\[
\frac{\phi(x_2)- \phi(x_1)}{x_2-x_1}\geq \frac{\phi(y_2)- \phi(y_1)}{y_2-y_1}.
\]
If $\phi$ is strictly concave, then equality is achieved if and only if $x_1 = y_1$ and $x_2 = y_2$.
\end{proposition}
\begin{proof}	
For $\phi$ concave, we have $\cs(x_1,x_2) \geq \cs(y_1,x_2) \geq \cs(y_1,y_2)$. Next, assume $\phi$ is strictly concave. Then equality holds in the first inequality if and only if $x_2 = y_2$, and in the second if and only if $x_1=y_1$, completing the proof.
\end{proof}

\printbibliography
\end{document}